%% file: main.tex
\newif\iflong\longtrue
\newif\ifanonymous\anonymousfalse
\iflong
\documentclass[11pt]{article}
\usepackage[utf8]{inputenc}
\usepackage[top=1in,left=1in,right=1in,bottom=1in]{geometry}
\else
\documentclass[sigconf]{acmart}
\fi
\usepackage{amsthm}
\usepackage{amsmath,amssymb}
\usepackage{hyperref}
\usepackage{relsize}
\usepackage[noend]{algpseudocode}
\usepackage[ruled,vlined]{algorithm2e}
\usepackage{multicol}
\usepackage{relsize}
\usepackage[most]{tcolorbox}
\usepackage{color}
\usepackage{thm-restate}
\usepackage{mathtools}
\usepackage{enumerate}
\usepackage{pgfplots}
\usepackage{subcaption}
\usepackage{graphicx}
\usepackage{booktabs}
\usepackage{xspace}
\pgfplotsset{
    compat=1.3,
    legend image code/.code={
        \draw [#1] (0cm,-0.1cm) rectangle (0.6cm,0.1cm);
    },
}
\usepackage[capitalize,nameinlink]{cleveref}

\theoremstyle{plain}
\newtheorem{theorem}{Theorem}[section]

\newtheorem{invariant}{Invariant}
\newtheorem{lemma}[theorem]{Lemma}
\newtheorem{corollary}[theorem]{Corollary}

\newtheorem{challenge}[lemma]{Challenge}
\newtheorem{observation}[theorem]{Observation}

\def\polylog{\operatorname{polylog}}

\input{macro}

\input{utils}

\iflong
\else
\usepackage[font=small,aboveskip=0pt,belowskip=0pt]{caption}
\usepackage[compact]{titlesec}
\titlespacing*{\section}{0pt}{3pt}{3pt}
\titlespacing*{\subsection}{0pt}{3pt}{2pt}
\titlespacing*{\subsubsection}{0pt}{3pt}{2pt}
\fi
\crefname{theorem}{Theorem}{Theorems}
\Crefname{lemma}{Lemma}{Lemmas}
\Crefname{claim}{Claim}{Claims}
\Crefname{observation}{Observation}{Observations}
\Crefname{algorithm}{Algorithm}{Algorithms}
\Crefname{myalgctr}{Algorithm}{Algorithms}
\Crefname{challenge}{Challenge}{Challenges}

\DeclarePairedDelimiter\floor{\lfloor}{\rfloor}

\algrenewcommand\algorithmicindent{1em}

\newcommand{\tO}{\widetilde{O}}

\newcommand{\mysize}{\fontsize{9}{10}}
\renewcommand{\deg}{deg}

\usepackage{xcolor}

\begin{document}
\sloppy
\title{Near-Optimal Distributed Implementations of Dynamic Algorithms for Symmetry Breaking Problems}

\iflong
\else

\begin{abstract}
    The field of dynamic graph algorithms aims at achieving a thorough
    understanding of real-world networks whose topology evolves with
    time. Traditionally, the focus has been on the classic sequential,
    centralized setting where the main quality measure of an algorithm is its
    update time, i.e.\ the time needed to restore the solution after each
    update. While real-life networks are very often distributed across multiple machines, the fundamental question of finding efficient
    \emph{dynamic, distributed graph algorithms} received little attention to
    date. The goal in this setting is to optimize both the \emph{round} and
    \emph{message} complexities incurred per update step, ideally achieving a
    message complexity that matches the centralized update time in $O(1)$
    (perhaps amortized) rounds.

    Toward initiating a systematic study of \emph{dynamic, distributed
    algorithms}, we study some of the most central symmetry-breaking problems: maximal
    independent set (MIS), maximal matching/(approx-) maximum cardinality matching (MM/MCM),
    and $(\Delta + 1)$-vertex coloring.
                        This paper focuses on
		dynamic, distributed algorithms that are \emph{deterministic},
    and in particular --- robust against an \emph{adaptive} adversary.
    Most of our focus is on the
            MIS algorithm,
    which achieves $O\left(m^{2/3}\log^2
    n\right)$ amortized messages in $O\left(\log^2 n\right)$ amortized rounds
    in the \textsc{Congest} model.
    Notably, the amortized message complexity of our algorithm matches the amortized update
            time of the best-known deterministic
    {\em centralized} MIS algorithm by Gupta and Khan [SOSA'21]
        up to a $\polylog n$ factor.
		The previous best deterministic distributed MIS algorithm, by
        Assadi et al.\ [STOC'18],
		uses $O(m^{3/4})$ amortized messages in $O(1)$ amortized rounds, i.e.,
        we achieve a polynomial improvement in the message complexity by a
        $\polylog n$ increase to the round complexity; moreover, the algorithm
        of Assadi et al. makes an implicit assumption that the network is
        connected at all times, which seems excessively strong when it comes to
        dynamic networks.
                                        Using techniques similar to the ones we developed for our
        MIS algorithm, we also provide
    deterministic algorithms
        for MM, approximate MCM and $(\Delta +
    1)$-vertex coloring whose message complexities match or nearly match the update times of the
    best centralized algorithms, while
    having either constant or $\polylog(n)$ round complexities.
\end{abstract}
\fi

\iflong
\ifanonymous
\author{Anonymous Author(s)}
\else
\author{Shiri Antaki\thanks{Tel Aviv University
    \href{mailto:shiriantaki@mail.tau.ac.il}{shiriantaki@mail.tau.ac.il}}
\and Quanquan C. Liu\thanks{Massachusetts Institute of Technology
    \href{mailto:quanquan@mit.edu}{quanquan@mit.edu}}
\and Shay Solomon \thanks{Tel Aviv University \href{mailto:
solo.shay@gmail.com}{
solo.shay@gmail.com}}}
\fi

\maketitle

\iflong
\begin{abstract}
    The field of dynamic graph algorithms aims at achieving a thorough
    understanding of real-world networks whose topology evolves with
    time. Traditionally, the focus has been on the classic sequential,
    centralized setting where the main quality measure of an algorithm is its
    update time, i.e.\ the time needed to restore the solution after each
    update. While real-life networks are very often distributed across multiple machines, the fundamental question of finding efficient
    \emph{dynamic, distributed graph algorithms} received little attention to
    date. The goal in this setting is to optimize both the \emph{round} and
    \emph{message} complexities incurred per update step, ideally achieving a
    message complexity that matches the centralized update time in $O(1)$
    (perhaps amortized) rounds.

    Toward initiating a systematic study of \emph{dynamic, distributed
    algorithms}, we study some of the most central symmetry-breaking problems: maximal
    independent set (MIS), maximal matching/(approx-) maximum cardinality matching (MM/MCM),
    and $(\Delta + 1)$-vertex coloring.
                        This paper focuses on
		dynamic, distributed algorithms that are \emph{deterministic},
    and in particular --- robust against an \emph{adaptive} adversary.
    Most of our focus is on our
            MIS algorithm,
    which achieves $O\left(m^{2/3}\log^2
    n\right)$ amortized messages in $O\left(\log^2 n\right)$ amortized rounds
    in the \textsc{Congest} model.
    Notably, the amortized message complexity of our algorithm matches the amortized update
            time of the best-known deterministic
    {\em centralized} MIS algorithm by Gupta and Khan [SOSA'21]
        up to a $\polylog n$ factor.
		The previous best deterministic distributed MIS algorithm, by
        Assadi et al.\ [STOC'18],
		uses $O(m^{3/4})$ amortized messages in $O(1)$ amortized rounds, i.e.,
        we achieve a polynomial improvement in the message complexity by a
        $\polylog n$ increase to the round complexity; moreover, the algorithm
        of Assadi et al. makes an implicit assumption that the network is
        connected at all times, which seems excessively strong when it comes to
        dynamic networks.
                                        Using techniques similar to the ones we developed for our
        MIS algorithm, we also provide
    deterministic algorithms
        for MM, approximate MCM and $(\Delta +
    1)$-vertex coloring whose message complexities match or nearly match the update times of the
    best centralized algorithms, while
    having either constant or $\polylog(n)$ round complexities.
\end{abstract}

\fi

\input{intro}
\input{tech-overview}

\input{prelims}
\input{MM-coloring}

\input{maximum-matching}

\input{main-alg}
\input{m-max-restart}
\input{main-analysis}
\input{m-avg}
\input{analysis-avg}

\appendix
\input{appendix}

\section*{Acknowledgements}

We are grateful to Boaz Patt-Shamir for useful comments on an earlier version of
this paper that helped improve its presentation.

\bibliographystyle{alpha}
\bibliography{ShayBib}

\end{document}

%% file: macro.tex
\newcommand{\congest}{\textsc{CONGEST}\xspace}

\newcommand{\hd}{high-degree\xspace}
\newcommand{\ld}{low-degree\xspace}
\newcommand{\mmax}{m_{max}}
\newcommand{\mavg}{m_{avg}}
\newcommand{\vld}{V_L}
\newcommand{\vhd}{V_H}
\newcommand{\alg}{\mathcal{A}}
\newcommand{\eps}{\varepsilon}
\newcommand{\mcur}{m_{cur}}

%% file: utils.tex
\DeclareMathOperator{\poly}{poly}

\definecolor{mygreen}{RGB}{20,140,80}
\definecolor{linkcolor}{RGB}{0,0,230}
\definecolor{mylightgray}{RGB}{230,230,230}
\definecolor{verylightgray}{RGB}{245,245,245}

\hypersetup{
     colorlinks=true,
     citecolor= blue,
     linkcolor= blue,
     urlcolor= blue
}
\newcommand{\myparagraph}[1]{\vspace{1.5pt}\noindent {\bf #1.}}

\newcounter{myalgctr}

\newtcolorbox{OuterBox}[1][]{

    breakable,
    enhanced,
    frame hidden,
    interior hidden,
    left=-5pt,
    right=-5pt,
    top=-5pt,
    float=p,
    boxsep=0pt,
    arc=0pt
#1}

\newtcolorbox{InnerBox}[1][]{

    enforce breakable,
    enhanced,
    colback=gray,
    colframe=white,
#1}

\newenvironment{tbox}{
\vspace{0.2cm}
\begin{tcolorbox}[width=\columnwidth,
                  enhanced,
                  boxsep=2pt,
                  left=1pt,
                  right=1pt,
                  top=4pt,
                  boxrule=1pt,
                  arc=0pt,
                  colback=white,
                  colframe=black,
	              breakable
                  ]

}{
\end{tcolorbox}
}

\newcommand{\tboxhrule}[0]{\vspace{0.1cm} {\color{black} \hrule} \vspace{0.2cm}}

\newenvironment{titledtbox}[1]{\begin{tbox}#1 \tboxhrule}{\end{tbox}}

%% file: intro.tex
\section{Introduction}\label{sec:intro}
Traditional graph algorithms process a \emph{static} graph on a single
(centralized) machine and are sequential;
thus, their runtime is at least linear in the graph size.
Even linear-time static algorithms, which are traditionally considered extremely fast
and optimal, are often inadequate for coping
with \emph{modern big data}, which dynamically changes and evolves at a
rapid pace. Often, such big data cannot be stored in one machine and hence
distributed methods are required to process it.
Efficiently coping with such
data is widely recognized as one of the most important
challenges of modern computation.
A \emph{dynamic} graph algorithm is one that efficiently deals with rapid
changes to the input graph, where
a common goal is to maintain a subgraph with some key property while the underlying graph changes over time.
A \emph{distributed} graph algorithm is one that  efficiently deals with graph
data stored in multiple machines, where the corresponding processors work
\emph{in parallel}
in order to achieve a common goal by communicating and coordinating their actions
via message passing.

The focus of dynamic graph algorithms, up until the past few years, has
almost exclusively been in the classic \emph{sequential, centralized} setting,
where the main quality measure is the algorithm's {\em update time}, i.e., the
time needed to update the graph structure of interest per update step.
Meanwhile, the focus of distributed algorithms has been mainly on the
\emph{static} setting, primarily on the \emph{round complexity} of {\em static tasks}.
Surprisingly, the fundamental question of \emph{distributing} known sequential
dynamic graph algorithms received very little attention to date. Most previous
works \cite{AOSS18,BKM19,BEG18,CDK20,CHK16,KG18,KS18,PPS16,PS16,LPR09} focused on small amortized round complexity.
A few works considered message complexity~\cite{AOSS18,PS16} but only for
certain problems.

Minimizing the number of messages is an important goal.
A small number of messages implies a small load on the communication links, which in some cases enables running multiple algorithms (or multiple instances of the same algorithm) concurrently and it can allow for pipeline implementation.
Moreover, this measure is useful when considering cases where one cares about the total work done (as captured by the number of messages sent), e.g.\ when the total network bandwidth is limited. We remark that many real-world systems are often bandwidth limited; examples of bandwidth limited systems include systems with poor wireless connections, an over-saturated network (many independent agents are on the same network, as in, e.g., a large company), or a mobile data network in a poorly connected area. On a low bandwidth network, an algorithm which uses less rounds but more messages in principle can actually take longer to finish its execution than an algorithm that uses less messages but more rounds.

We define the \emph{message-efficiency}
of our algorithms to be the ratio between the \emph{amortized
message complexity} per update of our algorithm
and the sequential amortized
update time of the best-known centralized algorithm. Note that although
we consider message-efficiency as a separate property, we still want our
algorithms to run in $\poly\log n$ or $O(1)$ rounds, as is standard.
Our goal is to design distributed algorithms with (nearly) constant message-efficiency, meaning the amortized message complexity asymptotically matches the amortized update time of the best centralized algorithm for the problem.
Of course, at the same time, we want to upper bound the amortized number of rounds of such
algorithms by (nearly) constant.
We allow $\poly \log n$ slacks in both the message and round complexities.

In this paper, we aim at initiating a systematic study of dynamic, distributed message-efficient algorithms
by considering a single edge update per step, as in the classic, sequential centralized setting. Even for this basic setting,
there are many challenges underlying the adaption of centralized algorithms to the
distributed setting.
Very recently, Censor-Hillel et al.\ \cite{CDK20} and Bamberger et al.\
\cite{BKM19} studied the question of simultaneously handling concurrent updates
in a distributed setting. However, their algorithms require $\Omega(m)$
amortized messages. Thus, it seems crucial to first thoroughly understand the
base case of a single update and only later extend to more general settings.
We focus on classic \emph{symmetry-breaking} problems:
maximal independent set (MIS), maximal matching (MM), $3/2$-maximum
cardinality matching (MCM), and $(\Delta
+ 1)$-vertex coloring.
Symmetry-breaking constitutes one of the most important
challenges in distributed computing, since in many distributed systems processors might be in the same state,
yet one must somehow break the symmetry to perform almost any nontrivial
computation.

One challenge in optimizing the message-efficiency
is that many centralized algorithms for these problems rely on
global variables and data structures, and they often rely on periodic \emph{global restarts}.
The idea behind a periodic global restart is for the algorithm to handle some number of updates ``lazily'', essentially without
changing the data structure, until a sufficient number of updates have been
accumulated.
Such centralized approaches are problematic in the distributed setting as each
individual vertex does not have access to global information.
Moreover, many centralized algorithms~\cite{AOSS18,DZ18,GuptaK18,KNNP20,NS13}
also
assume knowledge of the number of edges in the graph at any time, which does not lose generality in a centralized setting but is a very strong assumption to make in (dynamic) distributed settings --
such information can be acquired, but through potentially
expensive communication.
As a result, distributing centralized dynamic algorithms is a challenging task.

We shall restrict our attention to the standard $\congest$ model of communication
(\cite{PelB00}), where message size is bounded by $O(\log n)$ bits.
In particular, we use the most studied model of the distributed dynamic setting, the \emph{local
wakeup ($\congest$) model} (cf.\ \cite{AOSS18,CHK16,KS18,PPS16,PS16}),
where following an edge
update $(u,v)$, only the updated vertices $u$ and $v$ wake up.
The update procedure proceeds in fault-free synchronous rounds during which
every processor {\em that has been woken up} is allowed to exchange $O(\log n)$-bit messages
with its neighbors until finishing its execution, which differs from the
standard static setting in a crucial aspect: in
the static setting all the vertices are woken up at the outset and engage in the
algorithm,
whereas in the dynamic setting a vertex has to be woken up as part of the update
procedure---by receiving a message that  propagated from either $u$ or $v$---in
order to engage in the update procedure. In particular, to achieve good
message-efficiency, vertices cannot blindly participate in the update procedure
as even the size of the $2$-hop neighborhood of an updated vertex can
be large, leading to large message complexity;
this poses a highly nontrivial challenge for
algorithms which need to maintain some global invariant.

The round and message complexities can be viewed as the ``runtime'' and ``total work'' of the algorithm, respectively.
Our paper focuses on solving the dynamic MIS problem via a deterministic, \congest
algorithm that minimizes message complexity and number of rounds using new
techniques to resolve the issue surrounding an unknown number of edges in the entire
graph.
We then use similar techniques
to those used in our algorithm for solving MIS
to solve the other problems discussed in this paper, namely, MM,
$3/2$-approximate MCM, and
$(\Delta + 1)$-vertex coloring.
For MM, $3/2$-approximate MCM, and $(\Delta + 1)$-coloring, there exist
$O(\Delta)$-message, $O(1)$-round naive algorithms where each node queries its
immediate neighbors for the desired property (their color, whether they are in a
matching\dots etc.); the goal is beat these algorithms in number of messages in
the setting where $\Delta = \Omega(\sqrt{m})$,
where $\Delta$ is a fixed upper bound on the maximum degree in the
graph.
For MM, $(3/2)$-approximate MCM, and $(\Delta + 1)$-coloring,
we provide the first non-trivial $O(1)$-message-efficient
deterministic,
distributed algorithms that match the update times of $O(\sqrt{m})$ of the
best-known centralized algorithms~\cite{KNNP20,NS13}. Our algorithm runs in
$O(1)$ worst-case rounds for MM and $(\Delta + 1)$-coloring and $O(\log \Delta)$
worst-case rounds for $3/2$-approximate MCM.
Furthermore, our algorithms for these
problems are simple and, hence, practically implementable.

\myparagraph{Related Work for MIS}
We give a more comprehensive survey of the related work on MIS since this
problem has received more attention in the dynamic, distributed setting.
The MIS problem has been intensively studied over the years
since the celebrated works of \cite{AlonBI86,Linial87,Luby86} from the mid 1980s.
Most of the literature on the problem arguably revolves around parallel and distributed settings, perhaps due to the practical applications of MIS in these settings.

In recent years there has been a growing body of work on the problem of
\emph{maintaining} an MIS in dynamic networks
\cite{AOSS18,AOSS19,BDHSS19,CHK16,CZ19,DZ18,GuptaK18,KW13}.
Although most of this work has focused on the standard centralized, sequential
setting, Censor-Hillel et al.\ \cite{CHK16} gave a {randomized} algorithm for maintaining
an MIS against an \emph{oblivious adversary} in distributed dynamic
networks.\footnote{In the \emph{oblivious adversarial model} (see, e.g.,
\cite{CW77b,KKM13}), the adversary has complete knowledge of all the edges in
the graph and their arrival order, as well as of the algorithm, but does not
have access to the random bits used by the algorithm, and so cannot adapt its
updates in response to the random choices of the algorithm.}
The distributed algorithm of \cite{CHK16} achieves an (amortized) message
complexity of $\Omega(\Delta)$ and an \emph{expected} round complexity of
$O(1)$.
König and Wattenhofer \cite{KW13} gave an algorithm for maintaining an MIS which requires a constant number of rounds, but as opposed to \cite{CHK16}, makes the assumptions that a node gracefully leaves the network, and messages may have unbounded size. Also, the number of broadcasts that are done may be large.
Assadi et al.~\cite{AOSS18}
showed that their main centralized algorithm for MIS can be naturally adapted
into a deterministic distributed algorithm with amortized
message complexity $O(\min\{m^{3/4},\Delta\})$ and an amortized round complexity of $O(1)$.
However, this algorithm operates under the assumption that knowledge
of up-to-date estimates of $m$ is provided to all vertices; this leads to an
implicit assumption that the graph
remains connected throughout the course of the algorithm,
which is a strong assumption to make, especially in dynamic networks.

In the centralized, sequential setting,
the deterministic $O(m^{3/4})$ bound of Assadi et al.\ \cite{AOSS18} for general
graphs was improved to $O(m^{2/3})$ by Gupta and Khan \cite{GuptaK18} and
independently to $O(m^{2/3} \sqrt{\log m})$ by Du and Zhang
\cite{DZ18}.
Allowing randomization against an oblivious adversary, the update time in general graphs was reduced to
$O(\sqrt{m})$ \cite{DZ18}, further to $\tO(\min\{m^{1/3},\sqrt{n}\})$ \cite{AOSS19},
and ultimately to $\poly \log n$ \cite{BDHSS19,CZ19}.
To the best of our knowledge, the only distributed algorithms for the problem
are the two mentioned above \cite{AOSS18, CHK16}.
That said, it seems rather straightforward to distribute the randomized
algorithms with amortized update time $\polylog(n)$ \cite{BDHSS19,CZ19},  to
obtain a distributed algorithm with both message and round complexities bounded
by $O(\polylog(n))$.\footnote{We did not make any effort to verify this claim, as
    our   goal was to achieve an efficient {\em deterministic} algorithm, or at
least one that does not assume an {\em oblivious} adversary.} However, the
disadvantage of these randomized algorithms is that they crucially
require the oblivious adversary assumption.
While such an assumption might be fine in the centralized setting, in the
distributed setting it is easier for adversaries to corrupt links in the
network as well as to corrupt and/or eavesdrop on the messages sent through such links.
Thus, in such settings the oblivious
adversary assumption seems excessively strong and impractical.
In this paper, we shall restrict our attention
to deterministic algorithms, which are, in particular, robust against an adaptive adversary.
Prior to this work
there was no deterministic (or even randomized against a non-oblivious adversary)
distributed algorithm for MIS with a message complexity of $o(m^{3/4})$,
and moreover, it seemed highly unclear if the deterministic algorithms of \cite{DZ18,GuptaK18},
with amortized update time $\tilde O(m^{2/3})$ could be distributed efficiently.

\subsection{Our Contributions}
\myparagraph{Maximal independent set}
We present a deterministic distributed algorithm that achieves amortized
message and round complexities of $O(m^{2/3}\log^2{n})$ and $O(\log^2{n})$, respectively.
To this end, we reduce the problem of dynamically maintaining an MIS to that of statically computing it.
Our reduction builds on the aforementioned (centralized, sequential) algorithm of Gupta and Khan \cite{GuptaK18} in a nontrivial way.

\begin{theorem} \label{t1}
Equipped with a black-box static deterministic algorithm for computing an MIS within $T(n)$ rounds for any $n$-vertex distributed network, an MIS can be maintained \emph{deterministically} (in the local wakeup model) over any sequence of edge insertions and deletions that start from an empty distributed network on $n$ vertices,  within $O(T(n))$ \emph{amortized round complexity} and
	$O(m^{2/3} \cdot T(n))$ \emph{amortized message complexity},
	where $m$ denotes the dynamic number of edges.
\end{theorem}

Using the MIS algorithm of \cite{GGR20,RG20}, which runs in $O(\log^5 n)$ rounds, on top of the transformation of Theorem \ref{t1}, the amortized bounds on the round and message complexities of the resulting distributed algorithm are $O(\log^5 n)$ and $O(m^{2/3} \log^5 n)$, respectively.

While the black-box static MIS algorithm used in the transformation of Theorem \ref{t1} applies to arbitrary graphs,
bounded diameter graphs admit faster MIS algorithms.
Next, we strengthen the transformation of Theorem \ref{t1}, to achieve a diameter-sensitive transformation.
We stress that the algorithm returned as output of this transformation applies to any dynamic graph;
the restriction on the diameter is only for the black-box static MIS algorithm.
The black-box MIS algorithm should also satisfy another property; given as input an independent set $M' \subseteq V$
of the graph, the output MIS should be a superset of the input set $M'$; we shall call this an {\em input-respecting MIS}.
\begin{theorem} \label{t2}
Equipped with a black-box static deterministic algorithm for computing an
input-respecting MIS within $T'(n)$ rounds for any $n$-vertex distributed
network {\em with diameter at most 6},
an MIS can be maintained \emph{deterministically}
over any sequence of edge insertions and deletions that start from an empty network on $n$ vertices,  within $O(T'(n))$ \emph{amortized round complexity} and
	$O(m^{2/3} \cdot T'(n))$ \emph{amortized message complexity},
	where $m$ denotes the dynamic number of edges.
\footnote{For our purposes it suffices to
take a constant of 6, but any fixed constant $c \geq 6$ works.}
\end{theorem}

The MIS algorithm of \cite{CPS20} runs in $O(D \log^2 n)$ rounds in distributed graphs of diameter $D$.
We adapt the algorithm of \cite{CPS20} to return an input-respecting MIS.
Plugging the resulting MIS algorithm into the transformation of Theorem \ref{t2} yields:
\begin{corollary}\label{cor}
	Starting from an empty distributed network on $n$ vertices, an MIS can be maintained \emph{deterministically}
		over any sequence of edge insertions and deletions with $O(\log^2 n)$ \emph{amortized round complexity} and
	$O(m^{2/3} \log^2 n)$ \emph{amortized message complexity}.
\end{corollary}

Our algorithm uses unicast rather than broadcast messages,
which allows each processor to communicate differently with each of its
neighbors, and, more concretely,
to communicate with a subset of its neighbors---otherwise there is no hope to achieve a message complexity of $o(\Delta)$.

Prior to this work, the distributed algorithm of \cite{AOSS18} was the only one
providing a deterministic algorithm with $o(m)$ amortized message complexity. By
allowing the amortized round complexity to grow from constant to a small
polylogarithmic factor, we obtain a polynomial improvement in the message
complexity.
Perhaps more important than this improvement, in contrast to the work of
\cite{AOSS18} that relies on the network being connected at all times, our
algorithm does not make any assumptions on the network's connectivity. Assuming
that the network is always connected seems to be too much to ask for --- particularly
for dynamic networks. To cope efficiently with disconnected graphs, our
algorithm has to bypass several nontrivial challenges, some of which are
discussed next in~\cref{sec:tech-overview}.

\myparagraph{Other symmetry-breaking problems}
We also show the following results for MM, $(3/2)$-approximate MCM,
and $(\Delta+1)$-vertex coloring,
which were not known prior to this paper.
All of our results for these problems are $O(1)$-message-efficient, matching
the sequential running time of the best centralized deterministic
algorithms for these problems~\cite{KNNP20,NS13}.
\begin{theorem}\label{thm:mm}
    Starting from an empty distributed network on $n$ vertices, a maximal
    matching (MM) and a $(3/2)$-approximate maximum matching
    can be maintained \emph{deterministically} over any sequence
    of edge insertions and deletions with $O(1)$ and $O(\log \Delta)$ rounds,
    worst-case, respectively, and
    $O(\sqrt{m})$ worst-case and amortized message complexity, respectively.
\end{theorem}

\begin{theorem}\label{thm:coloring}
    Starting from an empty distributed network on $n$ vertices, a $(\Delta +
    1)$-vertex coloring can be maintained \emph{deterministically} over any
    sequence of edge insertions and deletions with $O(1)$ rounds, worst-case,
    and $O(\sqrt{m})$ worst-case message complexity.
\end{theorem}

\subsection{Paper Organization}
In~\cref{sec:tech-overview}, we provide a technical overview of our algorithms.
Then, in~\cref{sec:simple} we review the edge orientation technique and present
our algorithms for MM and $(\Delta+1)$-coloring.
In~\cref{sec:MCM}, we present our algorithm for $(3/2)$-approximate MCM.
Finally, in~\cref{sec:MIS}, we provide our new algorithm for MIS that improves
on~\cite{AOSS18}.

%% file: tech-overview.tex
\section{Proof Overview and Technical Challenges}\label{sec:tech-overview}

Throughout the paper we show that oftentimes centralized algorithms that require a {\em global} property such as the
current value of $m$ can instead use information obtained from each vertex's
{\em local} neighborhood to approximate the global property. Similarly, instead of maintaining {\em global} invariants, we settle for local, weaker variations of such invariants, and demonstrate that they are sufficiently strong for our purposes.
Our insights and techniques may be
applied more broadly to obtain deterministic, dynamic, distributed algorithms
beyond the specific problems we study in this paper.

We first consider the maximal matching and $(\Delta + 1)$-vertex coloring
problems, and our proposed algorithms for them can be viewed as a ``warm-up''
for our later algorithms. Our algorithms employ {\em dynamic edge orientations}
to achieve $O(\sqrt{m})$ message complexity (matching the centralized update
time), together with constant round complexity, without knowing the current number of
edges $m$ and without making any connectivity assumptions.
We demonstrate that dynamic edge orientations are useful for ensuring
that each vertex sends information to at most $O(\sqrt{m})$ neighbors per
update. By maintaining the \emph{inherently local} invariant that edges are
oriented from vertices with \ld to vertices with \hd, we can
achieve a maximum outdegree of $O\left(\sqrt{m}\right)$. Using this, we can make
sure that each vertex has complete information on all its incoming neighbors,
and only need to spend $O(\sqrt{m})$ messages to find the remaining information
on its out-neighbors.
More specifically,
for algorithms where vertices only need to know information about their direct
neighbors, we show that such an orientation is enough to ensure $O(\sqrt{m})$ message
complexity (and $O(1)$ round complexity).

We then proceed to the $3/2$-maximum cardinality matching (MCM) problem. To
achieve an efficient algorithm in this case, the idea of dynamic edge
orientations is insufficient because vertices need not only information about
their direct neighbors but also information about their $2$-hop neighborhood.
To ensure that the number of messages stays
$O(\sqrt{m})$, we made two key insights. First, using the same \hd, \ld
partitioning as for MM and $(\Delta + 1)$-coloring (cutoff of $\Theta(\sqrt{m})$ but
without knowing $m$),
we ensure that \emph{all high-degree vertices are always matched}. We show that
the following is true for each \emph{unmatched} \hd vertex $w$: either $w$ can
find a free (unmatched) neighbor in the first $\Theta(\sqrt{m})$ neighbors it queries, or $w$ can
find a neighbor that is matched to a \ld vertex (we call this the
\emph{surrogate}) in the first $\Theta(\sqrt{m})$
neighbors it queries.
This fact allows us to ensure that $w$ is always matched.
The second key insight is that searching $O(\sqrt{m})$ neighbors of any vertex
allows us to determine whether the vertex is \hd or \ld, and, importantly, we do
this without needing to know the value of $m$.
This insight allows
us to do the following for each vertex incident to a matching
edge that was deleted: we can search its neighbors by guessing the value of
$\sqrt{m}$ in $O(\log \Delta)$ attempts;
starting with one neighbor, we successively search two times the previous number
of neighbors until either we find a free neighbor or a surrogate (or we run out
of neighbors). Once we find a free neighbor or a surrogate, we are done since we
have successfully matched the vertex. By our observation above, it is sufficient
to search $\Theta(\sqrt{m})$ neighbors for each \hd vertex before finding a free
neighbor or a surrogate. Since we double the number of neighbors we search, we
overestimate the number of such neighbors by at most a factor of $2$ and so the number
of messages sent is $O(\sqrt{m})$ (dominated by the last guess). The number of
rounds is $O(\log \Delta)$ since we require this many guesses before guessing
$\Delta$.
These insights show that to approximate a \hd/\ld partition for the MCM
problem, it is enough to look at (part of) a vertex's $2$-hop neighborhood.

Our main result for this paper is
solving the dynamic, distributed MIS problem; hence, we spend the remainder of
this section on its challenges and solutions. We use the ideas we obtain from
the other problems combined with several new insights in order to achieve our
final algorithm. The main challenge we face is that we are no longer able to use
only the $2$-hop neighborhood in our algorithm; instead we must use (part of) the $6$-hop
neighborhood and we formulate a new \emph{restart} procedure (discussed a bit
later) for this purpose.

Our distributed algorithm for maintaining an MIS, as summarized in Corollary \ref{cor},
is a direct consequence of the reduction of Theorem \ref{t2}, which strengthens that of Theorem \ref{t1}.
The starting point of both our reductions is the centralized sequential
algorithm of \cite{GuptaK18}, hereafter {Algorithm [GK]}, with amortized update time
$O(m^{2/3})$. We start by giving a high-level overview of [GK].

[GK] dynamically partitions the vertex set $V$ into two subsets $V_H$ and $V_L$
of high-degree and low-degree vertices, respectively, according to a degree
cutoff of $m^{2/3}$,   giving priority to the vertices of $V_L$ to be in the MIS
over those of $V_H$. For dynamic graphs of maximum degree $\le \Delta$, there is
a deterministic algorithm, hereafter Algorithm $\alg$, for maintaining an MIS with
amortized update time $O(\Delta)$ \cite{AOSS18}.
The subgraph $G_L = G[V_L]$ of $G$ induced by $V_L$ has maximum degree $\le
m^{2/3}$, thus an MIS $M_L$ for $V_L$ can be maintained with update time
$O(m^{2/3})$ by applying Algorithm $\alg$ on $G_L$.
No vertex of $V_H$ that is incident to an MIS vertex of $V_L$ can be in the MIS.
Let the remaining vertices of $V_H$ which are not in the MIS and not adjacent to
any vertices in the MIS be $V^*_H$; let $G^*_H = G[V^*_H]$ be
the subgraph of $G$ induced by $V^*_H$.
Since $G^*_H$ contains at most $O(|V^*_H|^2) = O(|V_H|^2) = O(m^{2/3})$ edges, a
fresh MIS $M_H$ for $G^*_H$
can be computed in $O(m^{2/3})$ time following every edge update by running any
linear-time static MIS algorithm. Specifically, the algorithm used in [GK] is
the trivial one of iteratively picking an arbitrary $v \in V^*_H$
to add to $M_H$ and then removing $v$ and all adjacent vertices from $V^*_H$;
the procedure runs until $G^*_H$ is empty.
The output MIS, given by $M := M_L \cup M_H$, is thus maintained in $O(m^{2/3})$ update time.

We next highlight some of the challenges that we faced on the way to achieving
an efficient distributed implementation of [GK].
\begin{challenge} \label{c1}
The task of computing a fresh MIS $M_H$ on $G^*_H$ following every edge update cannot be distributed efficiently.
In fact, this task as is cannot be distributed {\em even inefficiently}.
\end{challenge}
First and foremost, computing an MIS on $G^*_H$ requires all the vertices in $V^*_H$ to wake up.
Alas, in the local wakeup model, only the updated vertices are woken up following an edge update,
hence waking up all the vertices in $V^*_H$ may be prohibitive if the diameter is large and even infeasible (if the graph is not connected).
Instead of computing an MIS on the entire $G^*_H$, we propose to apply it on a {\em carefully chosen} subgraph of $G^*_H$,
denoted here by $\tilde G_H$, where all the vertices of $\tilde G_H$ are at constant distance from the updated vertices.
Simply waking up all vertices at constant distance from the updated vertices in $O(1)$ rounds to participate in a static MIS distributed algorithm won't work, as an MIS computation over the corresponding subgraph will cause conflicts due to edges and vertices lying outside of that subgraph. Consequently, we need to first identify the vertices among those that will not trigger any further conflicts, and compute the subgraph $\tilde G_H$ induced on them only.
Moreover, for the reduction of Theorem \ref{t2}, we would need $\tilde G_H$ to have a bounded diameter.
This is impossible in general, since, even if the diameter of the entire graph is small, we have no control whatsoever on the diameter of $G^*_H$, let alone on the diameter of $\tilde G_H$.

To overcome this difficulty, we must add to $\tilde G_H$ some vertices of $V_L$---those lying on the shortest paths between the updated vertices and the vertices of $\tilde G_H$ that we wish to apply the static MIS algorithm on.
We then face another challenge: the resulting graph, denoted by $G'_H$, contains vertices of $V_L$ on the one hand, but on the other hand the static MIS algorithm that we apply on $G'_H$ must not affect the MIS $M_L$ for $V_L$, as that would blow up both the round and message complexities.
To this end, we need to apply an input-respecting MIS algorithm on top of $G'_H$, where the input independent set should contain all vertices of $V_L$ in $G'_H$ that belong to $M_L$. This, in turn, requires us to adapt the algorithm
of \cite{CPS20} to be input-respecting.

\begin{challenge} \label{c2}
Vertices (processors) do not know the number of edges in the graph, which dynamically changes.
\end{challenge}
The update times of various (centralized) dynamic graph algorithms depend polynomially on the dynamic number of edges $m$.
For example, in the context of symmetry breaking problems,
for MIS there is an $O(m^{3/4})$ update time algorithm \cite{AOSS18}
and $\tO(m^{2/3})$ update time algorithms \cite{DZ18,GuptaK18};
for maximal and $(3/2 + \eps)$-approximate maximum matching there are $O(\sqrt{m})$
and $O(m^{1/4})$ update time algorithms \cite{GP13,BS16,NS13}; and for
$(\Delta+1)$-coloring there is an $O(\sqrt{m})$ update time algorithm
\cite{DBLP:journals/corr/abs-1909-07854}.
In these algorithms and others, achieving an update time polynomial in $m$ usually requires knowledge of $m$ or some approximation of it, where this knowledge is directly translated into the dynamic maintenance of data structures and invariants with respect to $m$.
In particular, the standard way to cope with a dynamic number of edges $m$ in a centralized setting is to apply a global
``restart'' procedure every time the number of edges grows or shrinks by a factor of 2, where the role of a restart is not merely to compute the new number of edges, but also to rebuild the data structures with respect to the new value of $m$, thereby restoring the validity of the invariants.

In the local wakeup model, we shall assume that each update step has a running timestamp associated with it.
We believe this assumption, which was used before (cf.\ \cite{AOSS18}), should be acceptable in practice.
Moreover, a variant of our algorithm does not require the usage of timestamps,
with the caveat that its amortized message complexity is $O(m^{2/3}_{max}\log^2{n})$
(see~\cref{sec:restart} for details).
For brevity, here, we focus on our algorithm that uses timestamps.

 Importantly, only the updated vertices learn about the timestamp of the corresponding update.
{\em If the graph were always connected}, a global restart procedure could be
distributed in a straightforward way, at least in terms of computing the
up-to-date number of edges. In this case, every vertex would store the timestamp
$t_{last}$ of the last restart as well as the number of edges $m_{last}$ in the
graph at that time, and once an update step occurs with timestamp $t_{last} +
m_{last}/2$, one can compute the up-to-date number of edges $m$ in the graph
within $O(m)$ rounds and messages, and then broadcast this number as well as the
current timestamp to all vertices.
Such a restart procedure was employed by Assadi et al.\ \cite{AOSS18} for
distributing their centralized algorithm.
We believe that the assumption that the graph is always connected is too strong when dynamic graphs are involved.
Indeed, a standard assumption in many dynamic graph algorithms with amortized time bounds (including that of \cite{AOSS18}) is that the initial graph is empty.
Moreover, the underlying goal is to improve over the $\Delta = O(n)$ message
complexity of the algorithm of \cite{AOSS18},  which means that the relevant
regime of graph densities is $m = o(n^{3/2})$.
If the graph is not connected, however, it seems hopeless for vertices to learn
up-to-date number of edges in the graph. This is, in fact, a rather
generic problem, which arises naturally when attempting to distribute any
centralized dynamic graph algorithm that requires knowledge of $m$
\cite{BS16,DZ18,GuptaK18,GP13,DBLP:journals/corr/abs-1909-07854,NS13}.

When $m$ is not known by the vertices, we can no longer maintain the invariant
that the sets $V_L$ and $V_H$ consist of the low-degree and high-degree
vertices, respectively. Indeed, a global change in the number of edges also
changes the degree cutoff of $m^{2/3}$, hence it is possible for this invariant
to hold at some point and to be drastically violated by many vertices at a later
step. This implies that vertices of $V_L$ could have very high
degree, hence
applying a distributed implementation of $\alg$ on top of the subgraph $G_L$ of
$G$ induced by $V_L$ would blow up the message complexity; symmetrically,
vertices of $V_H$ could have very low degree, hence their number may be much
larger than $m^{1/3}$, in which case the subgraph $G'_H$ of $G$ on which a
static MIS computation is applied
may contain many more edges than $m^{2/3}$, which too may blow up the message complexity.
Overcoming this hurdle is not a minor technicality, it is in fact the crux of our distributed algorithm.
To this end we propose several new structural insights into the problem,
which enable us to replace the global invariants of the centralized algorithm of \cite{GuptaK18}
by local, carefully-tailored invariants
that can be maintained efficiently in the local wakeup model, to ultimately
achieve our result. The key idea behind our approach is that we can use
the local neighborhood (instead of the global neighborhood) of a vertex $v$ to
determine whether $v$ should be categorized as low-degree or high-degree.
Importantly, this local neighborhood approach does not require the graph to be
connected and we show that it is enough to obtain our desired message
complexity, together with small $\poly\log n$ round complexity.

%% file: prelims.tex
\section{Preliminaries}

Given an undirected graph $G = (V, E)$, which represents a distributed network, each vertex $v \in V$
has access to an adjacency list of its current neighbors.
We let $deg(v)$ denote the degree of $v$ or the size of $v$'s adjacency
list.
As commonly assumed, each vertex can distinguish its neighbors via unique IDs assigned to each vertex.

Following standard convention in dynamic graph algorithms, the graph is empty at the outset,
and there is a single edge update per step.
As mentioned, we consider the local wakeup model where, following an edge update (insertion or deletion),
only the two endpoints incident to that edge
\emph{wake up} and may initiate an update procedure.
Any other vertex remains asleep until receiving a message from a neighbor, at
which stage it can start participating in the update procedure (by exchanging
messages with its neighbors and performing other actions).
Computation proceeds in synchronous rounds; thus all vertices (those which are awake)
know when a round
of communication starts and ends,
and a vertex cannot respond to messages it received in the same round - it can
do so in the next round. Our algorithms operate in the \congest
model so messages have size
$O(\log n)$ bits.

In addition to considering worst-case rounds and messages per update,
in this paper, we also consider
\emph{amortized round complexity} (or {\em amortized update time}) and
\emph{amortized message complexity}, which bound the \emph{average} number of
communication rounds
and $O(\log n)$-bit messages sent, respectively, needed to update the solution per edge update, over a \emph{worst-case} sequence of updates.
In contrast to the static setting,  the amortized message complexity of distributed tasks could be sublinear in the graph size, which makes it natural to optimize both measures of amortized round and message complexities rather than just the former.

We consider three different notions of number of edges in the graph: $m_{max}$,
$\mavg$ and $m_{cur}$. $\mmax$ is the maximum number of edges in the graph over
all updates in the update sequence. $\mavg$ is the average number of updates in
the graph over the update sequence. Finally, $\mcur$ is the current number of
edges in the graph after the most recent update. When the particular definition
being used is obvious from context, we simply use $m$.

%% file: MM-coloring.tex
\section{Dynamic MM and \texorpdfstring{$(\Delta + 1)$}{}-Coloring via Dynamic Edge
Orientations}\label{sec:simple}

We begin our discussion of our dynamic, distributed algorithms for
symmetry-breaking problems with our simple warm-up algorithms for MM and
$(\Delta + 1)$-coloring.
We first present the edge orientation technique which is one way of ensuring
that vertices only need to send messages to a small enough subset of neighbors
upon an update. This technique is based on the
idea that for each vertex $v$, if the edge $(u,v)$ is oriented towards $v$, then
$v$ has all the relevant information about $u$. For every edge insertion or deletion
that causes
$u$ to change its properties, $u$ notifies all its outgoing neighbors.
In such a case, when $v$ needs to make a decision using all
its neighbors' properties, it only needs to ask its outgoing neighbors about
their properties to have a full picture of all its neighbors.
By using the edge orientation technique we can solve problems like dynamic, distributed
MM and $(\Delta + 1)$-coloring. All round and message complexities
in this section are in terms of
$\mcur$, the current number of edges in the graph. For convenience, we define
$m = \mcur$ in this section.

\subsection{Edge Orientation Algorithm}\label{para:orientation}
We orient the edges in the following way: provided an edge insertion $(u, v)$,
the edge is oriented from the vertex with lower degree to the vertex with higher
degree. If the degrees of the two vertices are equal, then the edge is oriented
arbitrarily.
Each vertex $v$ also maintains a number $p_v$ indicating its degree during the
last time it \emph{checked the orientation}
of all its adjacent edges
and \emph{reoriented} any of the checked edges.
Initially, we set $p_v = 1$ for all vertices.
If an edge insertion or edge deletion $(u, v)$
causes $v$'s degree to fall outside the range $[p_v/2, 2p_v]$, $v$ asks its
neighbors for their degree. For each neighbor $w$ of $v$, $v$ directs the edge $(w, v)$
from $w$ to $v$ if $\deg(w) \leq \deg(v)$. Otherwise, $v$ directs the edge $(v,
w)$ from $v$ to $w$. After performing the reorientation of edges, $v$ sets $p_v
= \deg(v)$.
The reoreintation of the edges is done gradually, $20$ edges per update step, which means that this reorientation process would be carried over during the course of the \emph{next} $\deg(v)/10$
updates incident to $v$. This means that we finish updating the changes before the next
set of updates causes $\deg(v)$ to fall outside the range $[p_v/2, 2p_v]$ again.\\
The pseudocode for this algorithm is given in~\cref{alg:orient} and~\cref{alg:reorient-edges}.

\iflong

\begin{algorithm}
	\caption{Orient\_edges$(u,v)$}\label{alg:orient}
    \SetAlgoLined
   	\LinesNumbered
    \KwResult{Orient the edges properly on an edge insertion or deletion.}
   	\textbf{Input: } An edge insertion or deletion $(u,v)$.\\
    \If{$\deg(v) > \deg(u)$}{Orient the edge from $u$ to $v$.}
    \Else{Orient the edge from $v$ to $u$.}
    \For{$z \in (u,v)$}
 	 {\If{$\deg(z) > 2p_z$ or $\deg(z) < p_z/2$}{Mark $z$ as in process of checking reorientation.\\$z$ updates $p_z = \deg(z)$.}
 	 \If{$z$ is marked as in process of checking reorientation}{Reorient\_edges($z$)}}
\end{algorithm}

\begin{algorithm}
	\caption{Reorient\_edges($z$)}\label{alg:reorient-edges}
    \SetAlgoLined
   	\LinesNumbered
    \KwResult{Reorient the edges adjacent to $z$ gradually}
   	\textbf{Input: } A vertex $z$ that needs to reorient its adjacent edges.\\
   	\If{all the edges adjacent to $z$ are marked as checked}{Mark $z$ as done. \\ \Return}
   	\For{$20$ adjacent edges $(z,x)$ of $z$ that still might need reorientation}{
     \If{$\deg(z) > \deg(x)$}{Orient the edge from $x$ to $z$.}
     \Else{Orient the edge from $z$ to $x$.}
     Mark $(z,x)$ as checked}
\end{algorithm}
\fi

The following invariant is crucial for the update algorithm.

\begin{invariant}\label{inv:edge-orientation}
For any edge $(u, v)$, if the edge is oriented from $u$ to $v$, then $\deg(u)
\leq 4\deg(v)$.
\end{invariant}

\iflong
We first prove that the invariant is always maintained under the above update
algorithm~\cref{alg:orient} by using the algorithms above and~\cref{obs:reorient}.

\begin{observation}\label{obs:reorient}
    An edge $(u,v)$ that was oriented from $u$ to $v$ is reoriented in~\cref{alg:reorient-edges} only if $\deg(v)>\deg(u)$.
\end{observation}

\begin{lemma}\label{lem:orient-update}
    \cref{inv:edge-orientation} is always maintained.
\end{lemma}

\begin{proof}
     We prove this lemma via induction on the update step. For the basis of the
     induction, we consider the graph at the beginning (before any edge update
     takes place), at which stage the graph is empty. Therefore, the first update would be an edge
     insertion $(u,v)$. If after this insertion the edge is oriented from $u$ to
     $v$ then by~\cref{alg:orient} $\deg(u) < \deg(v)$, so trivially
     $\deg(u)\leq 4\deg(v)$.
     Assume by induction that \cref{inv:edge-orientation} is maintained after update $i$, we now prove that \cref{inv:edge-orientation} is maintained after update $i+1$.
     Let $(u,v)$ be an edge in the graph after update $i+1$, such that the edge is oriented from $u$ to $v$.
     \begin{enumerate}
         \item If $(u,v)$ was inserted during update $i+1$: by~\cref{alg:orient}, the edge is oriented from $u$ to $v$ only if $\deg(v)>\deg(u)$ so trivially $\deg(u)\leq 4\deg(v)$. Note that there is no scenario in which $(u,v)$ was first oriented towards $v$ and then $u$ or $v$ would reorient this edge; this is because the orientation of an edge (both in the insertion and in reorientation) is determined according to the degrees of its endpoints, and since the degree of $u$ and $v$ did not change between the insertion and the reorientation, $(u,v)$ will not be reoriented.
         \item If $(u,v)$ was oriented from $v$ to $u$ after update $i$:
         Since $(u,v)$ is oriented from $u$ to $v$ after update $i+1$, it means that an edge reorientation was made. By~\cref{obs:reorient} this happens only if $\deg(v)>\deg(u)$.
         \item If $(u,v)$ was oriented from $u$ to $v$ after update $i$:
         by our induction hypothesis, after update $i$ it holds that $\deg(u)\leq 4\deg(v)$, if the degrees of $u$ or $v$ did not change during update $i+1$ then it still holds that $\deg(u)\leq 4\deg(v)$.
         If the degree of $u$ or $v$ changed, then an edge reorientation might occur, in that case, since $(u,v)$ is oriented from $u$ to $v$ after update $i+1$, $\deg(v)$ must be greater than $\deg(u)$.
         The only case that left to be handled is if the degree of $u$ or $v$ changed but an edge reorientation did not occur. Let $t$ be the time when the orientation of $(u, v)$ was last checked and let $\deg_t(u)$ and $\deg_t(v)$ be the degree of $u$ and $v$, respectively, at time $t$.
         Since $(u, v)$ is oriented from $u$ to $v$, it must be the case that $\deg_t(u) \leq \deg_t(v)$. Similarly, $\deg(v)$ and $\deg(u)$ must be within a factor of two of $\deg_t(v)$ and $\deg_t(u)$ respectively (otherwise, the edge would be checked at a later time than $t$). So $\deg(u) \leq 2\deg_t(u)$ and $\deg(v) \geq \deg_t(v)/2$. This means $\deg_t(u) \leq \deg_t(v)$ which implies that $\deg(u)/2 \leq \deg_t(u) \leq \deg_t(v) \leq 2\deg(v)$ which means that $\deg(u) \leq 4\deg(v)$.
     \end{enumerate}

\end{proof}

We prove the following property of the orientation maintained in this manner.
\else
    \cref{inv:edge-orientation} immediately allows us to prove the following
    property regarding the outdegree of any vertex in the graph.
\fi

\begin{lemma}\label{lem:outgoing}
Each vertex $v$ has at most $4\sqrt{m}$ outgoing neighbors.
\end{lemma}

\iflong
\begin{proof}
Suppose for contradiction that there exists a vertex $v$ with more than $4\sqrt{m}$ outgoing neighbors. Let
$D$ be the degree of $v$. Then, by~\cref{inv:edge-orientation}, each outgoing
neighbor of
$v$ has degree at least $D/4$. By our assumption, because $v$ has $> 4\sqrt{m}$
outgoing neighbors, $D = \deg(v) > 4\sqrt{m}$. This means that each outgoing
neighbor of $v$ has degree $> \sqrt{m}$. Thus the sum of the degrees in the graph is greater than
$4\sqrt{m} \cdot \sqrt{m} = 4m$. Since there are at most $m$ edges, the sum of the
degrees cannot be greater than $2m$, a contradiction. Hence, each vertex $v$ has at
most $4\sqrt{m}$ outgoing neighbors.
\end{proof}

Now, every time a vertex $v$ changes its state, $v$ updates all its outgoing neighbors about this change.
When the orientation of an edge $(u,v)$ updates,
$v$ must send its state to $u$ if the edge
is now
oriented towards $u$; otherwise, $u$ must send its state to $v$.
For the $(\Delta + 1)$-coloring problem we define the state of a vertex as its color and for the maximal matching problem we define the state of a vertex to be either free or matched.
This will guarantee that vertices would know the color of their incoming neighbors or which of their incoming neighbors are matched.
We will see later that by using~\cref{alg:orient} and preserving the state of the outgoing neighbors as described above, we can solve those problems efficiently.

\begin{lemma}\label{lem:incoming-neighbors-property}
    Each vertex $u$ has a complete information on the states of all its incoming neighbors after each update.
\end{lemma}

\begin{proof}
    We prove this lemma via induction on the update step. For the induction
    basis, we consider the graph at the beginning (before any edge update takes
    place), at which stage the graph is empty and therefore the condition is
    trivially maintained.
    Assume by induction that all vertices have complete information on the states of their incoming neighbors after update $i$.
    Let $u$ be an arbitrary vertex in the graph, we'll prove that $u$ has complete information on the state of its incoming neighbors after update $i+1$.
    For each incoming neighbor $v$ of $u$, if $v$ was an incoming neighbor of $u$ after update $i$, then by our explanation above, if the state of $v$ had been changed during update $i+1$, $v$ must have updated $u$ about it.
    Otherwise, by our induction hypothesis, $u$ has the state of $v$ after update $i$ and since it didn't change during update $i+1$, $u$ still has an updated information about $v$.
    If $v$ wasn't an incoming neighbor of $u$ after update $i$, it means that it became an incoming neighbor of $u$ during update $i+1$. This could happen if either $(u,v)$ was inserted during update $i+1$ or if the orientation of $(u,v)$ was changed during this update. According to the explanation above, in both cases $v$ updates $u$ about its state, so $u$ must have the current state of $v$. As this holds for each incoming neighbor $v$ of $u$, the lemma follows.

\end{proof}
\fi

\begin{lemma}\label{lem:message-round-complexity-orient}
    Given an edge insertion or deletion $(u, v)$,~\cref{alg:orient} requires $O(1)$
    message complexity and $O(1)$ round complexity, both worst-case.
\end{lemma}

\iflong
\begin{proof}
    Whenever the degree of a vertex $u$ falls out of the range $[p_u/2, 2p_u]$,
    we must check all the edges incident to $u$ and reorient if necessary.
    This naively costs
    $O(\deg(u))$ messages.

    To achieve $O(1)$ worst-case message complexity, we make the following procedure. When $\deg(u)$ falls outside the range
    $[p_u/2, 2p_u]$, we first update $\deg(u)$ to be the new degree. Then, we
    reorient $20$ edges incident to $u$ during each update incident to $u$, such that after $p_u/10$ updates incident
    to $u$, all the edges have been reoriented if necessary.
    The next $p_u/10$ updates incident to $u$ cannot increase the outdegree of $u$ by
    more than $p_u/10$, thus, the number of outgoing neighbors of $u$ is still
    $O(\sqrt{m})$. Futhermore, the next $p_u/10$ updates cannot cause $\deg(u)$ to fall
    outside the range $[p_u/4, 4p_u]$ and so another reorientation update cannot
    occur before the current update has been processed. Finally, since each edge
    update is incident to two vertices, we charge at most $40$ reorientations to the
    update (if both endpoints are performing reorientations).
    As for the number of rounds, in each update $u$ uses at most eight rounds of communication, performing the following
    tasks (as described in~\cref{alg:orient} and in the above explanation about the state of a vertex):
    \begin{enumerate}
        \item Check the degree of $v$ if the edge $(u,v)$ was inserted and receive it.
        \item If $\deg(u)$ falls outside the range $[p_u/2, 2p_u]$:
        \begin{enumerate}
            \item Ask $20$ neighbors their degrees.
            \item Receive the sent degrees.
            \item Reorient the edges if necessary.
        \end{enumerate}
        \item For new incoming neighbors (due to insertion or reorientation):
        \begin{enumerate}
            \item Ask its new incoming neighbors to send their state.
            \item Receive the sent states.
            \item Reorient the edges if necessary.
        \end{enumerate}
    \end{enumerate}

\end{proof}
\fi
\subsection{Distributed \texorpdfstring{$(\Delta + 1)$}{}-Coloring}

In this section, we maintain a $(\Delta + 1)$-coloring in the graph using the
edge orientation technique we presented in~\cref{para:orientation}.
First, given an edge update, we run~\cref{alg:orient}.
For each edge $(u,v)$ that was oriented from $u$ to $v$ and is reoriented (such that it is now oriented from $v$ to $u$), $v$ sends its color to $u$.
Therefore, the following  property is maintained: for every edge $(u,v)$ that is oriented from $u$
towards $v$, $v$ knows the current color
of $u$.
The algorithm works as follows: given an edge insertion $(u, v)$ between two vertices $u$ and $v$ with the
same color, pick one of the two vertices arbitrarily, w.l.o.g. $u$, to recolor.
$u$ queries all its outgoing neighbors for their colors. Then, $u$
picks a color that does not conflict with any of its neighbors. This can be done since $u$ already has complete information about all its incoming neighbors, so by asking its outgoing neighbors it receives information about all its neighbors' colors.  Finally, $u$
informs all its outgoing neighbors of its new color. All of $u$'s outgoing
neighbors store $u$'s color.
\iflong
\begin{algorithm}
\caption{$(\Delta + 1)$-Coloring}\label{alg:coloring}
    \SetAlgoLined
    \LinesNumbered
    \KwResult{Maintains a $(\Delta + 1)$-coloring in the graph upon edge insertion or deletion}
    \textbf{Input: } An edge insertion or deletion $(u,v)$.\\
    Call~\cref{alg:orient} on $(u,v)$.\\
    \If{there were edge flips}{\For{an edge $(z,x)$ that was flipped, where $z \in \{u,v\}$}{
    \If{the edge is oriented towards $x$}{$z$ updates $x$ about its color}
    \Else{$x$ updates $z$ about its color}}}
    \If{$(u,v)$ is an inserted edge}{\If{$u$ and $v$ have the same color}{
    pick a vertex from $\{u,v\}$ arbitrarily to recolor (w.l.o.g. $u$ is picked)\\
    \For{each outgoing neighbor $w$ of $u$}{$u$ asks $w$ for its color\\ $w$ responds to $u$ with its color}
    recall that $u$ has complete information of all colors of its incoming neighbors (~\cref{lem:incoming-neighbors-property} ) and now has complete information about all its neighbors\\\ $u$ picks a different color from all its neighbors\\
    \For{each outgoing neighbor $w$ of $u$}{$u$ informs $w$ about its new color}}}
\end{algorithm}
\fi

Because~\cref{alg:coloring} uses the edge orientation technique
presented above,
\cref{inv:edge-orientation} and~\cref{lem:incoming-neighbors-property} are
maintained after each edge insertion and deletion and the algorithm has the same
message and round complexity.

\begin{theorem}\label{lem:message-round-complexity-coloring}
    Starting from an empty graph, there exists a dynamic,
    distributed \congest $(\Delta + 1)$-coloring algorithm (\cref{alg:coloring}) that
    maintains a valid coloring of the distributed network in $O(\sqrt{m})$ message complexity
    and $O(1)$ round complexity, both worst-case.
\end{theorem}

\iflong
\begin{proof}
    By~\cref{lem:message-round-complexity-orient}, reorienting the edges if
    necessary costs $O(1)$ worst-case messages and at most $O(1)$ rounds.
    Now we show that whenever $u$ is recolored, it sends at most $O(\sqrt{m})$
    messages and~\cref{alg:coloring} uses $O(1)$ rounds. When $u$ is recolored, it sends its color
    only to its outgoing neighbors. By~\cref{lem:outgoing}, $u$ has at most
    $4\sqrt{m}$ outgoing neighbors. Thus, $u$ requires at most $O(\sqrt{m})$
    messages to send its color to its outgoing neighbors. This only requires one
    round of communication.

    Finally, we can show that the number of
    messages is $O(\sqrt{m})$ and number of rounds is $O(1)$ for each edge
    insertion or deletion. Following any edge deletion, no vertices need to be
    recolored. A vertex $v$ would only need to check all its edges if its degree
    decreases beyond $p_v/2$. Then,
    by~\cref{lem:message-round-complexity-orient},
    this results in $O(1)$ worst-case messages and $O(1)$ rounds per
    update. Following any edge insertion, $(u, v)$ where $u$ and $v$ have the same
    color, one of the two must recolor itself. As proved above, $u$ or $v$
    require $O(\sqrt{m})$ messages to send their new color in $O(1)$ rounds. If
    $u$ and $v$ have different colors, neither of them needs to recolor itself.
\end{proof}
\fi

\subsection{Maximal Matching}

We shall follow the above approach for $(\Delta+1)$-coloring
to maintain a maximal matching under
edge updates in $O(\sqrt{m})$ messages and $O(1)$ rounds, both worst case, per update.
For any edge $(u, v)$ that is oriented from $u$ to $v$, $v$ knows whether $u$ is
currently matched or free. Then, given any update we first perform the edge orientation
algorithm. (Any edge reorientation $(v, u)$ causes $v$ to send a message to $u$
indicating whether it is matched or free.)
For an edge insertion, no additional updates need to be made.
For an edge deletion $(u, v)$, if the deleted edge was an edge in the matching,
then we do the following for $u$ and after that for $v$; we next describe how to
handle $u$ but $v$ should be handled in the same way. $u$ first checks whether any of its incoming neighbors are free. If any are free, $u$ arbitrarily picks such an incoming neighbor to match with. If no incoming neighbors
of $u$ are free, $u$ asks its outgoing neighbors whether they are free (and its outgoing neighbors sends back their answer). If any are
free, $u$ arbitrarily picks a neighbor to be matched with. This algorithm allows us to achieve the same message and round
complexity as our $(\Delta + 1)$-coloring algorithm. The pseudocode for our
algorithm is provided in~\cref{alg:matching}.

\begin{algorithm}
\caption{Maximal Matching}\label{alg:matching}
    \footnotesize
    \SetAlgoLined
    \LinesNumbered
    \KwResult{Maintains a maximal matching in the graph following any edge update}
    \textbf{Input:} An edge insertion or deletion $(u,v)$.\\
    Call~\cref{alg:orient} on $(u,v)$.\\
    \If{there were edge flips}{\For{an edge $(z,x)$ that was flipped, where $z \in \{u,v\}$}{
    \If{the edge is oriented towards $x$}{$z$ updates $x$ about its current status (whether it is matched or free)}
    \Else{$x$ updates $z$ about its current status}}}
    \If{$(u,v)$ is a deleted edge}{\If{$u$ and $v$ were matched}{First do the following algorithm on $u$ and then on $v$ (to avoid the situation which both of them choose the same neighbor to be matched with)\\
    \For{$z \in \{u,v\}$}{
    $z$ informs all its outgoing neighbors that it is unmatched\\
    \If{there exists an incoming neighbor $x$ of $z$ that is unmatched}{match $z$ with $x$\\ $x$ informs all its outgoing neighbors that it is matched\\
    $z$ informs all its outgoing neighbors that it is matched}
    \Else{
    \For{each outgoing neighbor $w$ of $z$}{$z$ asks $w$ if it is unmatched\\ $w$ responds to $z$ if it's unmatched}\If{there exists an unmatched outgoing neighbor  $w$ of $z$}{match $z$ with $w$\\
    $w$ informs all its outgoing neighbors that it is matched\\
    $z$ informs all its outgoing neighbors that it is matched}}}}}

\end{algorithm}

\begin{theorem}\label{lem:message-round-complexity-mm}
    Starting from an empty graph, there exists a dynamic,
    distributed \congest MM algorithm (\cref{alg:matching}) that
    maintains a maximal matching in the distributed network in $O(\sqrt{m})$ message complexity
    and $O(1)$ round complexity, both worst-case.
\end{theorem}
\iflong
\begin{proof}
    For an edge update $(u,v)$, we handle $u$ and $v$ separately. We next show
    how to handle $u$, but $v$ should be handled in the same way.
    If $(u,v)$ was deleted and $u$ and $v$ were matched to each other: $u$ has
    complete information about its incoming neighbors
    by~\cref{lem:incoming-neighbors-property}, and, therefore, it can check if one
    of them is free. If a free incoming neighbor exists, it costs $O(1)$ rounds
    and messages to matched between $u$ and this neighbor. After that, notifying
    all the outgoing neighbors of $u$ about this change costs $O(1)$ rounds and
    $O(\sqrt{m})$ messages since $u$ has most $4\sqrt{m}$ outgoing neighbors
    by~\cref{lem:outgoing}.
    If $u$ does not have an unmatched incoming neighbor, then checking all of
    $u$'s outgoing neighbors and choosing a match, if one exists, also costs $O(1)$
    rounds and $O(\sqrt{m})$ messages.
    If $u$ and $v$ weren't matched to each other, or if $(u,v)$ was inserted, there is no need to look for a match to $u$ or $v$.
    The edge insertion or deletion might cause reorientation of edges;
    by~\cref{lem:message-round-complexity-orient}, reorienting the edges, if
    necessary, costs at most $O(1)$ rounds and messages.

\end{proof}
\fi

%% file: maximum-matching.tex
\section{\texorpdfstring{$3/2$-Approximate Maximum Cardinality Matching}{}}\label{sec:MCM}

In this section, we provide a distributed algorithm that uses $O(\sqrt{m})$ messages and $O(\log \Delta)$ rounds in the CONGEST model
to maintain a $3/2$-approximate maximum cardinality matching (MCM) in the input graph under edge insertions and deletions.
Our algorithm is
based on the sequential algorithm of~\cite{NS13}. The main challenge we must surmount in the distributed
setting is the fact that vertices do not know $\mcur$, the current number of edges in the graph.
Unfortunately, unlike the case with MM and $(\Delta + 1)$-coloring, it is no
longer sufficient to bound the number of outgoing edges of a vertex using the
edge orientation algorithm. Thus, in this section, we show a technique to
differentiate a \ld vertex from a \hd vertex using information obtained from the
$2$-hop neighborhood of a vertex without knowledge of the global $\mcur$.
As in~\cref{sec:simple}, we define for convenience $m = \mcur$. The message
complexity is given in terms of $\mavg$ since we use amortized complexity in
this case. The round complexity is given in terms of $\mcur$ since it is
given in worst case complexity.

Let $G=(V,E)$ be an arbitrary graph and let $M$ be an arbitrary matching for $G$. The edges of $M$ are called \emph{matched} edges and the \emph{unmatched} edges are the remaining edges $E \setminus M$.
An \emph{augmenting path} with respect to $M$ is a path whose edges alternate between edges in
$M$ and edges in $E \setminus
M$ which starts and ends on different \emph{free} vertices (i.e.\ vertices which are not
matched).
It is well-known that if $G$ does not have an augmenting path
of length $3$, then $M$
is a $3/2$-approximate MCM \cite{HK73}.
The natural algorithmic idea is to exclude all augmenting paths of length at most $3$ from the graph.

Ideally, we would like to use the edge orientation technique to efficiently maintain a $3/2$-approximate MCM as used for MM and $(\Delta + 1)$-coloring.
We know how to use the edge orientation
technique to efficiently maintain a maximal matching. However, to determine whether or not there is an augmenting
path of length $3$ in the graph, and if so to find it, the vertices adjacent to the edge update must
have updated information not only about their neighbors, but also about their
neighbors' neighbors, which makes the edge orientation technique insufficient for
solving this problem.
To find an augmenting path, if exists,
we first partition (using the procedure explained below) the vertices into \ld and \hd vertices based on the
threshold of $\Theta(\sqrt{m})$ (\hd vertices have degree greater than $\sqrt{2m}$ and \ld vertices
have degree less or equal to $\sqrt{2m}$).
We ensure
that \hd vertices $w$ would not go through all their neighbors
each round to find an augmenting path.
We do this by using an
algorithm that finds a \emph{surrogate} for each \hd vertex $w$ that becomes free. A \emph{surrogate} is a vertex $v'$ that is matched to a neighbor $v$ of $w$, such that the degree of $v'$ is at most $\sqrt{2m}$.
The key observation that was made in \cite{NS13} is the following: \emph{each unmatched \hd vertex can always
find a free neighbor or surrogate in the first
$O(\sqrt{m})$ neighbors it queries}.
Given this fact, we are able to look only at as many neighbors
of $w$ as necessary to find a surrogate. However, since we are working in the distributed setting we don't know the current value of $m$.
Therefore, the algorithm works as follows:
for a vertex $v$ that became free during an edge update,
in each round we ``guess'' the
number of neighbors we would like to check by successively \emph{doubling our number
of neighbors to check} (starting with $1$ neighbor), until we find a free
neighbor or surrogate or there
are no more neighbors to check.
Later on we show that this algorithm won't check
more than $O(\sqrt{m})$ neighbors although the algorithm may not have a good estimation on the value of $m$.
Using this algorithm we maintain the fact that high degree vertex
will \emph{always} be matched after an update that is incident to
it.
Furthermore, assuming there were no augmenting paths of length $3$ before the update,
after the update only vertices that are incident to the edge update or vertices
that became free during the update can be part of new augmenting paths of length
$3$.

For a vertex that incident to an edge update, if it is matched after the update then it cannot be the start or the end
of an augmenting path of length $3$, this is due to the fact that an augmenting path must start and end with unmatched vertices.
This means that a high-degree vertex that is incident to an edge update
will not be the start or the end of any augmenting path during this update unless it
becomes low-degree. Hence, we only need to look
for augmenting paths that start at free low-degree vertices.
This leads to a natural procedure: for a free low-degree vertex $u$, we need to
look at all its matched neighbors, e.g.\ $v$, and ask $v$ whether they have a
mate $v'$ which
has a neighbor, e.g.\ $w$, that is free.
For any such path $(u, v, v', w)$, we remove from the matching the edge $(v,v')$ and add instead the edges $(u, v)$ and $(v', w)$. It is easy to verify (e.g.\ \cite{NS13}) that this eliminates all augmenting paths that
start or end at $u$.

We describe this algorithm in detail below. This algorithm allows us
to obtain our desired result:

\begin{theorem}\label{lem:message-round-complexity-mcm}
    Starting from an empty graph, there exists a dynamic,
    distributed \congest $(3/2)$-approximate MCM algorithm
    (\cref{alg:maximum-matching-deletion}) that
    maintains a $(3/2)$-approximate MCM in the graph in $O(\sqrt{\mavg})$ amortized message complexity
    and $O(\log \Delta)$ worst-case round complexity.
\end{theorem}

\iflong
\myparagraph{Data Structures} Each vertex $v$ maintains the following information
in its local memory:

\begin{enumerate}
    \item A variable $mate_v$ that keeps its mate in the match (if it has one, otherwise it keeps $\varnothing$).
    \item A variable $d_v$ that keeps its degree.
    \item A linked list $N_v$ that maintains all its neighbors.
    \item A counter $f_v$ of the number of its free neighbors.
    \item A linked list $F_v$ that maintains all its free neighbors.
\end{enumerate}
\begin{algorithm}
\caption{$3/2$-approximate MCM: Edge insertion $(u, v)$}\label{alg:maximum-matching-insertion}
    \SetAlgoLined
    \LinesNumbered
    \KwResult{A $3/2$-approximate MCM in the graph.}
    \textbf{Input: } An edge insertion $(u, v)$.\\
        \If{$u$ and $v$ are both free}{Match($u,v$).\\ Add edge $(u,v)$ to the matching.\\
    $u$ notifies all its neighbors that it is matched.\\
    $v$ notifies all its neighbors that it is matched.\\}
    \ElseIf{$u$ is free}{Find\_Mate($u$,$v$)}\If{$v$ is free}{Find\_Mate($v$,$u$)}
\end{algorithm}

\begin{algorithm}
\caption{$3/2$-approximate MCM: Edge deletion $(u, v)$}\label{alg:maximum-matching-deletion}
    \SetAlgoLined
    \small
    \LinesNumbered
    \KwResult{$3/2$-approximate MCM in the graph.}
        \textbf{Input: } An edge deletion $(u, v)$.\\
            \For{$z\in\{u,v\}$}{
    \If{$u$ and $v$ were matched to each other}{
        \If{$z$ has at least one free neighbor}{
            $z$ chooses a free neighbor $w$ arbitrarily.\\
            Match($z,w$).\\ Add edge $(z,w)$ to the matching.\\
            $w$ notifies all its neighbors that it is matched.\\}
        {\Else{//Makes a call to~\cref{alg:maximum-matching-surrogate}.\\$w'$ = Surrogate($z$)\\
        \If{$w' \neq \varnothing$\\ //In this case, it is guaranteed that $z$ is matched and $w'$ is free.\\}{Match\_Surrogate($w'$)\\}
        \Else{$z$ notifies all its neighbors that it is free.\\
    $z$ changes its mate: $mate_z \leftarrow \varnothing$.\\Aug-path($z$)}}}}
        \ElseIf{$z$ is free\\ //We want to make $z$ a matched vertex in case it is \hd.\\}{$w'$ = Surrogate($z$)\label{line:find-surrogate}\\ \If{$w' \neq
            \varnothing$}{
                Match\_Surrogate($w'$)\\
                $z$ informs all its neighbors that it is matched.\\
        }}}
        \end{algorithm}

\begin{algorithm}
    \caption{Find\_Mate$(u, v)$}\label{alg:find-mate}
    \SetAlgoLined
    \LinesNumbered
    \KwResult{Find a mate to vertex $u$ if it exists.}
    \textbf{Input: } A vertex $u$ that is free
    and a vertex $v$ that is $u$'s neighbor and is matched.\\
    $u$ asks $v$ for $v$'s mate. $v' \leftarrow mate_v$.\\
        \If{$v'$ has a free neighbor which is not $u$\label{line:mate-aug-path}}{$v'$ chooses a free neighbor
    $x$ arbitrarily.\\ Match$(u,v)$.\\ Match$(v',x)$.\\ $u$ notifies all its
neighbors that it is matched.\\ $x$ notifies all its neighbors that it is
matched.} \Else{//Makes a call to~\cref{alg:maximum-matching-surrogate}.\\$w'$ = Surrogate($u$)\label{line:surrogate-mate}.\\
    \If{$w' \neq \varnothing$\\ //In this case, $u$ is matched and $w'$ is free.\\}
    {$u$ notifies all its neighbors that it is matched.\\
Match\_Surrogate($w'$)}}
\end{algorithm}

\begin{algorithm}
\caption{Match$(u,v)$}\label{alg:maximum-matching-match}
    \SetAlgoLined
    \LinesNumbered
    \KwResult{Match between two vertices $u$ and $v$}
    \textbf{Input: } Two vertices $u$ and $v$ such that the edge $(u,v)$ is in the graph ($u$ and $v$ might be free or matched before the call to this procedure, so updating their neighbors that they became matched is taken care of in the procedures that make the call to this algorithm).\\
    update the mate of $u$ to be $v$\\
    update the mate of $v$ to be $u$\\
\end{algorithm}

\begin{algorithm}
    \caption{Aug-path$(u)$}\label{alg:aug-path}
    \SetAlgoLined
    \LinesNumbered
    \KwResult{Find an augmenting path starting from $u$ if it exists.}
    \textbf{Input: } A vertex $u$ that is free and does not
    have a free neighbor.\\
    $u$ sends a message to all of its neighbors, notifying that it is looking for an
    augmenting path.\\
    \For{each vertex $w$ that received a message from $u$ in the previous round}
    {$w$ sends a message to its mate $w'$ and asks if it has a free neighbor.\\
    \If{$w'$ replies back that it has a free neighbor}
    {$w$ sends a message to $u$ that $w'$ is an option for an augmenting path.}}
    \If{$u$ receives back an option for an augmenting path}{$u$ chooses one of the options arbitrarily (say $w'$) and sends back to $w$ (the neighbor of $w'$ that is connected to $u$) that $w'$ is chosen.\\
    Match($u,w$)\\
    $u$ notifies all its neighbors that it is matched.\\
    $w'$ chooses a free neighbor $x$ arbitrarily.\\
    Match($w',x$)\\ $x$ notifies all its neighbors that it is matched.}

\end{algorithm}

\begin{algorithm}
\caption{Surrogate$(u)$}\label{alg:maximum-matching-surrogate}
    \SetAlgoLined
    \LinesNumbered
    \KwResult{Find a surrogate for $u$ if exists}
    \textbf{Input: } A vertex $u$\\
    Initialize $i=1$\\
    \While{$u$ did not find a surrogate and $i\le 2\log(deg(u))$}{
    $u$ sends a message to $\floor*{\sqrt{2^i}}$ arbitrary neighbors, notifying that it looks for a surrogate.\\
    \For{each vertex $w$ that received a message from $u$ in the previous round}{$w$ sends a message to its mate $w'$ and asks for its degree.\\  \If{$w'$ replies back that $deg(w') \le \floor*{\sqrt{2^i}}$}{$w$ sends a message to $u$ that $w'$ is a candidate for a surrogate}}
    \If{$u$ receives back a candidate for a surrogate}{$u$ chooses one of the candidates arbitrarily (say $w'$) and sends back to $w$ (the neighbor of $w'$ that is connected to $u$) that $w'$ is chosen\\
    Match($u,w$)\\
    $w'$ notifies all its neighbors that it is free.\\
    $w'$ changes its mate: $mate_{w'} \leftarrow \varnothing$.\\
    //In this case the output is a surrogate $w'$ which is now free.\\
    \Return $w'$}
    \Else{$i = i + 1$}}
    \Return $\varnothing$
\end{algorithm}

\begin{observation}
After~\cref{alg:maximum-matching-surrogate} is called on $u$, if the return value of this algorithm is a surrogate $w'$ of $u$, then $u$ is matched and $w'$ is free.
\end{observation}

\begin{algorithm}
\caption{Match\_Surrogate$(u')$}\label{alg:match-surrogate}
    \SetAlgoLined
    \LinesNumbered
    \KwResult{Find a match for $u'$ if exists.}
    \textbf{Input: } A vertex $u'$ which is a surrogate to a vertex $u$, in particular $u'$ is free.\\
    \If{$u'$ has a free neighbor}{$u'$ chooses a free neighbor $w'$ arbitrarily\\ Match ($u',w'$)\\ $w'$ notifies all its neighbors that it is matched\\}{\Else{Aug-path($u'$)}\label{line:find-aug-deletion}}
\end{algorithm}

\subsection{Analysis}

\begin{lemma}\label{obs:surrogate}
Assume~\cref{alg:maximum-matching-surrogate} is called with $u$ as its input (i.e., for finding surrogate($u$)). If $u$ is not matched after the call to this algorithm then $deg(u) \le \sqrt{2m}$.
\end{lemma}

\begin{proof}
As proved in \cite{NS13}, if a vertex $u$ has more than $\sqrt{2m}$ neighbors,
then at least one of them has a mate that can serve as a surrogate. A surrogate
of $u$ is a vertex $w'$ that has degree $\deg(w') \le \sqrt{2m}$ and is a mate of
one of $u$'s neighbors.
Assume by contradiction that $u$ has no surrogate but has more than $\sqrt{2m}$
neighbors. Since each of $u$'s neighbors has a different mate, there are more
than $\sqrt{2m}$ vertices that are mates of neighbors of $u$. By our assumption,
none of them can serve as a surrogate; by the definition of a surrogate, it
means that all of them have degree $> \sqrt{2m}$. Hence, the sum of the degrees
of all these vertices exceeds $\sqrt{2m} \cdot \sqrt{2m} = 2m$. Since the sum of
the degrees in the graph is $2m$ we get a contradiction.
We say that vertex $w'$ is a \textit{potential surrogate of $u$} if $w'$ is the
mate of some neighbor $w$ of $u$.
In~\cref{alg:maximum-matching-surrogate} we check on each iteration $i$ of the
loop for each potential surrogate $w'$, if $\deg(w') \le \floor*{\sqrt{2^i}}$. If
we go through all of $u$'s potential surrogates and did not find a surrogate, it
means that $\deg(u) \le \sqrt{2m}$, otherwise, we have at least $\log(2m)$
iterations of the loop since the loop goes up to iteration $2\log(\deg(u))$ which
    in this case is at least $2\log(\sqrt{2m}) = \log(2m)$. At iteration
    $\floor{\log(2m)}$ we would check if $\deg(w') \le
    \floor*{\sqrt{2^{\log(2m)}}} = \sqrt{2m}$ for each potential surrogate $w'$,
    and by our proof above it is not possible that all of the potential
    surrogates would have degree greater than $\sqrt{2m}$. If $\deg(u) >
    \sqrt{2m}$, we must find a surrogate after at most $\log(2m)$ iterations, as
    explained above. Therefore, we get that if $u$ is not matched after this
    algorithm then $\deg(u) \le \sqrt{2m}$.
\end{proof}

\begin{invariant}\label{inv:high-degree}
For a vertex $u$ that is adjacent to an edge update $(u,v)$, if $\deg(u) > \sqrt{2m}$ then after the update $u$ must be matched.
\end{invariant}

\begin{lemma}\label{lem:high-deg-matched}
\cref{inv:high-degree} is maintained by our algorithm after each update.
\end{lemma}

\begin{proof}
Suppose an edge update $(u,v)$ occurred such that $\deg(u)>\sqrt{2m}$.
\paragraph{If $(u,v)$ was deleted:} In that case if $u$ is matched and its mate
is not $v$ then after the update it is still matched and there is nothing to do.
If $u$ was matched to $v$, then we must find it a new mate. If $u$ has a free
neighbor $w$ then we match it with $w$. If not then we find a surrogate to $u$
by calling~\cref{alg:maximum-matching-surrogate}.
By~\cref{obs:surrogate}, the algorithm finds a surrogate (and therefore a
match) to $u$ since $\deg(u) > \sqrt{2m}$. Thus, $u$ becomes matched after the
update.
\paragraph{If $(u,v)$ was inserted:}
If $u$ is matched then the update doesn't affect it and it stays matched.
If $u$ is free, then we first check if $u$ can be
matched to $v$ (either if $v$ is free or if $u$ and $v$ are part of an
augmenting path). If $u$ can't be matched to $v$
then~\cref{alg:maximum-matching-surrogate} is called on $u$.
By~\cref{obs:surrogate}, since $\deg(u) > \sqrt{2m}$, $u$ will be matched after
running this algorithm.
\\\\
In both cases, $u$ is matched after the update.
\end{proof}

\begin{lemma}\label{lem:aug-path}
If~\cref{alg:aug-path} (Aug\_path($u$)) is called with a vertex $u$ as its input, then if there
exists an augmenting path of length $3$ such that $u$ is one of its
endpoints,~\cref{alg:aug-path} will find it.
\end{lemma}

\begin{proof}
The algorithm is called on a vertex $u$ that became free during some update and
has no free neighbors. The only way for an augmenting path of length $3$
to start or end with $u$ is for one of $u$'s matched
neighbors $v$ to have a mate $v'$ which has a free neighbor $w$.
During~\cref{alg:aug-path},
$u$ checks with all of its matched neighbors' mates to see if any has a free
neighbor. If such a free vertex $w$ exists, then $u$ and $v$ become
matched together and $v'$ and $w$ become matched together ($v$ is no longer matched to $v'$).
Since $u$ is now matched, it cannot be the start or the end of any augmenting
paths of length $3$.
\end{proof}

\begin{lemma}
    \cref{alg:maximum-matching-deletion} and~\cref{alg:maximum-matching-insertion} maintain
    a $(3/2)$-approximate maximum matching in the graph after each edge update.
\end{lemma}

\begin{proof}
We prove this lemma via induction on the update step that after each update step, there are no augmenting paths of length of at most $3$ with respect to the current matching, which implies that the maintained matching is a $(3/2)$-approximate maximum matching. In the base case, the graph is empty
and a $(3/2)$-approximate maximum matching exists trivially.
We assume as our induction hypothesis that at the $k$-th update step, there are no augmenting paths of length
$3$. We now prove this for the $(k + 1)$-st step.

Suppose that there is an edge update between the vertices $u,v$ for the $(k +
1)$-st step.
\paragraph{If $(u,v)$ was deleted from the graph:}
Since the algorithm is symmetric with respect to $u$ and $v$, we prove
that the lemma holds for $u$ and so it would also hold for $v$.
\begin{enumerate}
    \item If $u$ and $v$ weren't matched: then the deletion of the edge wouldn't
        affect the matching. However, since we need to maintain the invariant
        (\cref{inv:high-degree})
        that if $\deg(u) > \sqrt{2m}$ then it should be matched after the update,
        we search for a surrogate to $u$ as described
        in~\cref{line:find-surrogate} in~\cref{alg:maximum-matching-deletion}.
        If $\deg(u) > \sqrt{2m}$ then by~\cref{obs:surrogate} a surrogate must exist; otherwise, if $u$ does not have any free neighbor
        and it does not have a surrogate, then we know that $\deg(u) \leq
        \sqrt{2m}$.
        If we find a surrogate $u'$, then $u'$ becoming free
        might create augmenting paths of length $3$ (where $u'$
        is one of the endpoints in the augmenting paths).
        Therefore, if $u'$ doesn't have a free
        neighbor, an augmenting path that starts with $u'$ is looked for as described in~\cref{alg:match-surrogate} in~\cref{line:find-aug-deletion}.
        By~\cref{lem:aug-path}, if such a path exists then we will find it.
        Note that if $u'$ has a free neighbor $w$, they become matched. After
        becoming matched, $u'$ can't be part of an augmenting path since by the
        maximality of the matching, $w$ could not have had a free neighbor
        (otherwise, $w$ would have been matched).
    \item If $u$ and $v$ were matched, by~\cref{alg:maximum-matching-deletion},
        we are looking for a surrogate or a mate for $u$.
    \begin{enumerate}
        \item If $u$ did not find any mate, by
            \cref{obs:surrogate}, $\deg(u) \leq \sqrt{2m}$
            and $u$ doesn't have a free neighbor.
            We still might find an augmenting path that $u$ is one of
            its endpoints, so~\cref{alg:aug-path} is called with
            $u$ as its input. By~\cref{lem:aug-path}, if there is an augmenting
            path of length at most $3$ where $u$ is one of its
            endpoints,~\cref{alg:aug-path} must find it. Note that any new
            augmenting path must contain one of $u$ or $v$ as an endpoint, no other augmenting path exists by our induction
            hypothesis (if $u$ or $v$ is not part of the augmenting path, then
            the augmenting path existed before this update which is not
        possible, and if $u$ and/or $v$ are part of an augmenting path but are
        not one the endpoints of the path,
        then the maximality of the matching is not preserved).
        \item If we do find a match to $u$ then there is no augmenting path that
        $u$ is part of. First, assume that $u$ had a free neighbor $w$ and therefore was matched to it. In this case, an augmentig path that includes $u$ can exist only if $w$ has a free neighbor and this is not possible by the maximality of the matching.
        If $u$ doesn't have a free neighbor then $u$ must have a surrogate $u'$. In~\cref{alg:match-surrogate} we check if $u'$ has a free
        neighbor and if so we match $u'$ with its free neighbor. If it doesn't find a free neighbor then we check for an augmenting path that starts or ends with $u'$. Again, by~\cref{lem:aug-path}, if there is an augmenting path of length $3$ that has $u'$ as one of its endpoints,~\cref{alg:aug-path} must find it. Since
        the only vertex that became and remained free during this update is $u'$, there can't be another augmenting path of length $3$ in the graph by our induction hypothesis.
        Again, if $u$ and/or $u'$ are matched, then there can't be an augmenting path that $u$ or $u'$ are part of.
    \end{enumerate}
    We have shown that following any edge deletion,
    there is no augmenting path
    of length at most $3$ in the graph.
\end{enumerate}
\paragraph{If $(u,v)$ was inserted to the graph:}
\begin{enumerate}
    \item If $u$ and $v$ are both free then we match them. After that we can't have any augmenting path of length at most $3$ that contains $(u, v)$ in the middle since that would imply that both $u$ and $v$ had a free neighbor before the update; By our induction hypothesis, we know we had a maximal matching before the update, thus this is not a possible case. Also, trivially, $u$ and $v$ can't be one of an augmenting path's endpoints, since they are matched.
    \item If one of $u$ and $v$ is free, w.l.o.g. assume $u$: then a new augmenting path might occur between $u$, $v$ and $v$'s mate (say $v'$).
    By~\cref{alg:find-mate}, we first check if $v'$ has a
    free neighbor; if so, we have an augmenting path and therefore remove the edge $(v,v')$ from the matching, match $u$
    with $v$ and match $v'$ with its free neighbor. After that we don't have
    any other option for an augmenting path of length $3$, since if we do, then
    this path already existed before this update, in contradiction to our induction hypothesis.
    Since there was a maximal matching before the update, $u$ doesn't have a
    free neighbor after the update and the maximality is preserved. However,
    after this update $u$ might be a high degree vertex, therefore we look for a
    surrogate for $u$ as described in~\cref{alg:find-mate}
    in~\cref{line:surrogate-mate}. The surrogate $w'$ of $u$ (if exists), might
    have a free neighbor or might be the start of an augmenting path of length
    $3$, therefore, we first match $w'$ with a free neighbor if exists and if
    not we call~\cref{alg:aug-path} on $w'$. By~\cref{lem:aug-path} if there is
    such a path, we will find it. As we noted above, if $w'$ is matched to a free neighbor, it can't be part of an augmenting path by the maximality of the matching.
\end{enumerate}
    We get that also in the insertion case there is no augmenting path
    of length $3$ in the graph.
    We completed the induction step, which means that there are no augmenting paths of length at most $3$, and therefore we maintain a valid $(3/2)$ approximate maximum matching in the graph after each update.
\end{proof}

\begin{observation}\label{obs:surrogate-deg}
    If $u'$ is a surrogate of a vertex $u$ that was found
    during~\cref{alg:maximum-matching-surrogate}, then $\deg(u') \le \sqrt{2m}$.
\end{observation}

\begin{proof}
By~\cref{alg:maximum-matching-surrogate} if we found the surrogate $u'$ at
iteration $i$, then its degree must be $\deg(u') \le \floor*{\sqrt{2^i}}$. As
explained in the proof of~\cref{obs:surrogate}, if we find a surrogate we must
find it after at most $\log(2m)$ iterations, which means that $i \le \log(2m)$
which implies that $\deg(u') \le \floor*{\sqrt{2^{\log(2m)}}} = \sqrt{2m}$.
\end{proof}

\begin{observation}\label{obs:aug-path}
    Only vertices $v$ with degree $\deg(v) \le \sqrt{2m}$ can be an input
    to~\cref{alg:aug-path} as used by our algorithm.
\end{observation}

\begin{proof}
In~\cref{alg:find-mate} or~\cref{alg:match-surrogate} we
call~\cref{alg:aug-path} only with a surrogate as its input. The surrogate $u'$ can have degree at most $\sqrt{2m}$ by ~\cref{obs:surrogate-deg}.
In~\cref{alg:maximum-matching-deletion} we first call~\cref{alg:maximum-matching-surrogate} on a vertex $u$ that was incident to the edge deletion, and only if a surrogate for $u$ hasn't been found, we call~\cref{alg:aug-path}. By~\cref{obs:surrogate}, a vertex that is not matched after the call of~\cref{alg:maximum-matching-surrogate} must have degree $\le \sqrt{2m}$.
\end{proof}

\begin{observation}\label{obs:iterations}
    In~\cref{alg:maximum-matching-surrogate}, the while loop must end after at most $2\log(\Delta)$ iterations.
\end{observation}

\begin{proof}
In every iteration that we didn't find a surrogate for a vertex $u$, $i$ is incremented by $1$; the loop terminates either when we find a surrogate or $i$ reaches $2\log(\deg(u)) \le 2\log(\Delta)$, so at most $O(\log{\Delta})$ iterations of the while loop are done.
\end{proof}

We are now ready to prove ~\cref{lem:message-round-complexity-mcm}.
\begin{proof}
We will show that each procedure takes $O(\sqrt{\mavg})$ amortized messages and at
most $O(\log{\Delta})$ rounds in the worst case.
\paragraph{Finding a surrogate (\cref{alg:maximum-matching-surrogate}):}
The while loop in the algorithm takes at most $O(\log{\Delta})$ rounds: by~\cref{obs:iterations} the number of iterations of the while loop is at most $O(\log{\Delta})$.
In each iteration there are at most $O(1)$ rounds, so the total number of rounds for the loop is $O(\log{\Delta})$.
The rest of the algorithm takes at most $O(1)$ rounds and therefore the total number of rounds in the worst case is $O(\log{\Delta})$.
As for the number of messages, in iteration $i$ the number of messages is at
most $O(\sqrt{2^i})$. Overall, since, in each iteration, $i$ increments by $1$
and the number of iterations is at most $2\log(\Delta) \le 2\log(m)$, the total
number of messages in the while loop is $\sum_{i=1}^{\log{m}}\sqrt{2^i} =
O(\sqrt{m})$. In the rest of the algorithm the number of messages is at most
$O(1)$ and therefore the total number of messages is $O(\sqrt{m})$.

\paragraph{Finding an augmenting path (\cref{alg:aug-path}):}
It is easy to see that all the operations in the algorithm take $O(1)$ rounds.
It is easy to verify that the number of messages is $O(\deg(u)+\deg(x))$ where $u$ is the input of the
algorithm and $x$ is the free vertex that is part of the augmenting path, if such path exists. By~\cref{obs:aug-path}, the degree of a vertex that is an input to
this algorithm is at most $O(\sqrt{m})$, so the total number of messages is
$O(\sqrt{m} + \deg(x))$. We have no control over the degree of $x$, hence we
next bound the \emph{amortized} number of messages sent due to this algorithm.
If $x$ has $\deg(x) \le c\sqrt{2m}$ for any constant $c$ then we are done;
henceforth, assume that $\deg(x)>c\sqrt{2m}$ for some constant $c>1$.
First, note that $x$ was free before this
algorithm and became matched
during this algorithm. By~\cref{inv:high-degree} we know that a vertex with
degree $> \sqrt{2m}$ that is adjacent to an edge update would be matched after
the update, therefore if $x$ has degree $> \sqrt{2m}$
and is free at the beginning of the update, it had degree $\le \sqrt{2m}$ after
processing the previous update adjacent to it.
Observe that $x$ can change its high-degree/low-degree designation
without being adjacent to an edge update only if the number of
the edges $m$ in the graph changes (some of the updates might be adjacent to $x$, but its degree during these updates is at most $\sqrt{2m}$, otherwise, by~\cref{inv:high-degree}, $x$ is supposed to be matched).
For any vertex $x$ where the degree of $x$ changes from $\sqrt{2m}$ to $> c\sqrt{2m}$ without edge updates adjacent to it,
at least $\Omega(m)$ updates need to occur.
After those $\Omega(m)$ updates, the number of vertices like $x$, that became
high degree without edge updates adjacent to them and remained free is at most
$\sqrt{2m}$ (this is because the number of vertices with degree $> \sqrt{2m}$ in
the graph is at most $\sqrt{2m}$), so if the cost of handling each such vertex
is $O(m)$, the total cost of handling all of these vertices is $O(m\sqrt{m})$ messages.
However, after $\Omega(m)$ updates, we can get a total of $\Omega(m\sqrt{m})$ coins by charging $\Theta(\sqrt{m})$ extra coins in each update.
Therefore, we charge extra $O(\sqrt{m})$ coins in each update, and get an amortized bound of $O(\sqrt{m})$ on the number of messages.

\paragraph{Handling edge insertion $(u,v)$
(\cref{alg:maximum-matching-insertion}):} If $u$ and $v$ are both free then
matching
them takes $O(1)$ rounds and messages and notifying all of their neighbors about
their
new status takes $O(1)$ rounds and $O(\deg(u)+\deg(v))$ messages. If only one of
them is free, w.l.o.g. $u$, we first check if there is an augmenting path that
includes $u$ and $v$ (as described in~\cref{alg:find-mate}
in~\cref{line:mate-aug-path}). To check if such augmenting path exists costs
$O(1)$ rounds and messages. If such augmenting path exists, we must match $u$
with $v$ and the mate of $v$ with a free vertex that we found (denote by $x$).
That costs $O(1)$ rounds and $O(\deg(u)+\deg(x))$ messages. If there isn't an
augmenting path that $u$ and $v$ are part of, we then try to find a surrogate to
$u$, which costs $O(\log{\Delta})$ rounds and $O(\sqrt{m})$ messages as analyzed
above. If
$u$ has a surrogate $u'$ then notifying all the neighbors of $u$ that it is
matched costs $O(1)$ rounds and $O(\deg(u))$ messages. Then, $u'$ also checks if
it has a free neighbor $w'$; if so, matching them and notifying all the
neighbors of $w'$ costs $O(1)$ rounds and $O(\deg(w'))$ messages. If $u'$ doesn't have a
free neighbor then we call~\cref{alg:aug-path} with $u'$ as its input; by~\cref{obs:surrogate-deg} and as
analyzed above this costs $O(1)$ rounds and $O(\sqrt{m})$ amortized messages.

Overall, the total cost of this algorithm is at most $O(\log{\Delta})$ rounds and
$O(\deg(u)+\deg(v)+\deg(x)+\deg(w')+\sqrt{m})$ messages.
Since $u, v, x, w'$ were free before this
algorithm and became matched
during this algorithm, we can analyze the running time of this algorithm as analyzed above (in the analysis of~\cref{alg:aug-path}), and get that the total cost of this algorithm would be at most
$O(\log{\Delta})$ rounds and $O(m)$ messages, and if we amortized the messages over all
updates, it would result in amortized $O(\sqrt{m})$ messages.

\paragraph{Handling edge deletion $(u,v)$ (\cref{alg:maximum-matching-deletion}):}
We analyze the round and message complexities of this algorithm with respect to one of the endpoints of this edge, w.l.o.g $u$; the argument for $v$ is symmetric.
If $u$ and $v$ were matched, then checking if $u$ has a free neighbor and,
if so, matching it costs $O(1)$ rounds and messages. $u$'s new match $w$
must update all its neighbors and that costs $O(\deg(w))$ messages.
If $u$ doesn't have a free neighbor then finding a surrogate to $u$ costs
$O(\log{\Delta})$ rounds and $O(\sqrt{m})$ messages. If $u$ has a surrogate $z'$,
then finding a free neighbor $w'$ to $z'$ and match it would cost $O(1)$ rounds
and messages. Notifying all neighbors of $w'$ that $w'$ is matched costs
$O(\deg(w'))$ messages. If $z'$ doesn't have a free neighbor then we
call~\cref{alg:aug-path} on $z'$; this costs $O(1)$ rounds and $O(\sqrt{m})$ amortized
messages as analyzed above. If $u$ doesn't have a surrogate then notifying all of $u$'s neighbors
that it is free costs $O(1)$ rounds and $O(\sqrt{m})$ messages using~\cref{obs:surrogate}.
Then, calling~\cref{alg:aug-path} on $u$ costs $O(1)$ rounds and $O(\sqrt{m})$ amortized
messages (using the analysis above of~\cref{alg:aug-path}).
If $u$ and $v$ weren't matched but $u$ was free, we do a similar process as
analyzed above, and the total cost of this process would be $O(\log{\Delta})$ rounds
and $O(\deg(w') + \sqrt{m})$ messages.

Overall, the total cost of this algorithm is at most $O(\log{\Delta})$ rounds and
$O(\deg(w) + \deg(w') + \sqrt{m})$ messages. Note that $w$ and $w'$ are both free before the update and became matched during the update.
As analyzed in the cost of~\cref{alg:aug-path}, the cost of
those vertices can be amortized to $O(\sqrt{m})$ messages per update, and
therefore the total cost of this algorithm is at most $O(\log{\Delta})$ in the worst case rounds and
$O(\sqrt{m})$ amortized messages.
\end{proof}
\fi

%% file: main-alg.tex
\section{Maximal Independent Set}\label{sec:MIS}
Building on the techniques used in the previous sections, we now present our main
algorithm for dynamic, distributed MIS. Similar to our algorithm for MCM, our
algorithm for MIS also needs to partition vertices into \ld and \hd vertices
without knowledge of $m$. However, instead of looking at the two-hop
neighborhood as in our algorithm for MCM, our algorithm for MIS instead looks at
a subset of vertices in the \emph{$6$-hop} neighborhood to determine whether
each updated vertex is \ld or \hd. Since the procedure for determining this
information is much more complicated, we dedicate an entire section to
its explanation (\cref{sec:restart}).

Our MIS algorithm is first analyzed with respect to the \emph{maximum number of edges},
denoted by $\mmax$,
that exist in the graph at any point in time, and later analyzed with respect to the
\emph{average number of edges}, $\mavg$. 
The latter analysis is much more challenging, since we do not assume that the
graph is connected,
hence vertices do not have up-to-date estimates of the current number $m$ of edges.
To carry out this analysis, we assume that edge updates are tagged with a global timestamp;
only the two vertices adjacent to the edge update can read the timestamp.
We assume that the timestamps are given by a global running number, where each timestamp designates the number of the current update step. Of course, since the update sequence can be arbitrarily long, the update steps (and timestamps) could become prohibitively large, so the system is allowed to periodically reset the timestamps in order to make sure that each timestamp can be represented via $O(\log(n))$ bits. We do not address here the technical details behind such periodic resets of  timestamps, as this may change from one system to another; this optimization lies within the responsibility of the designers of such systems.
(The same assumption was made also in the work of \cite{AOSS18}, even though
their work also made an additional connectivity assumption.)

In this section, we present our main deterministic algorithm that maintains
an MIS under edge insertions and deletions.
Our deterministic distributed algorithm is inspired at a high-level
by the sequential algorithm given in~\cite{GuptaK18} that partitions
the vertices into \emph{high-degree} and \emph{low-degree} vertices,
while giving priority to the low-degree vertices to be included in the MIS.
Because each vertex does not have access to the precise value of $m$ (neither
$\mmax$ nor $\mavg$), instead we partition the vertices based on information
from the local neighborhood of each vertex. We update the degree designations of
each vertex as its local neighborhood changes. We provide our full algorithm in
the next few sections.

\subsection{High-Degree/Low-Degree Partitioning}\label{sec:main}

First, we provide some necessary characteristics and invariants maintained by
vertices in our algorithm.

\paragraph{High-Degree and Low-Degree Vertices}
Our algorithm ensures that all vertices in the input graph are always
partitioned into
\emph{\hd} and \emph{\ld} vertices, denoted by $\vhd$
and $\vld$, respectively.
We first provide an intuitive definition of these two concepts
and then provide the formal definition as it relates to our algorithm.

\paragraph{Intuitive definition based on~\cite{GuptaK18}}
In the sequential algorithm of~\cite{GuptaK18},
each vertex $v \in V$ that has degree $deg(v) > m^{2/3}$ in $G$
(where $m$ is the current number of edges in the graph) is
labeled as a \hd vertex; otherwise,
it is labeled as a \ld vertex. Since, in the centralized
setting, the algorithm has access to
the current number of edges
in the graph, such a definition suffices for this setting.

However, in the distributed setting, a vertex does not know the current
number of edges in the graph, as described in~\cref{c2}.
Thus, we must perform the partition
differently in the distributed model.

\paragraph{High-degree/\ld partitioning in the distributed setting}
We provide the partitioning algorithm in terms of $\mmax$ in this section
and update it accordingly for $\mavg$ in \cref{sec:m-avg}.

Each vertex $v$ stores a counter $deg'(v)$ indicating the degree it thinks is the
\emph{movement bound} for moving from $\vld$ to $\vhd$
or vice versa.
If a vertex $v$ has degree $deg(v) > deg'(v)$
then it labels itself \hd; otherwise,
it labels itself \ld.
We initialize all $deg'(v)$ values to be
$deg'(v) = 2$ for all $v \in V$ in the beginning
when $deg(v) = 0$ and the graph is empty. We describe
how to update $deg'(v)$ in~\cref{sec:high-degree}.
Intuitively, we show that $deg'(v)$
approximates $\mmax^{2/3}$.
When we consider $\mavg$ in~\cref{sec:m-avg},
$deg'(v)$ will not always approximate $\mavg^{2/3}$ and therefore we instead
use the timestamp of the current update to determine when we need to update a
vertex's degree designation (see
\cref{sec:m-avg} for more details).

Let $G_H = (V_H, E_H)$ be the subgraph induced by
the \hd vertices in $G$ and $G_L = (V_L, E_L)$ be
the subgraph induced by the \ld vertices in $G$.
In our distributed setting $G_H$ and $G_L$ are
composed of vertices which currently (in the present round) consider
themselves to be \hd and \ld, respectively.

\subsection{Algorithm Overview}
We consider low-degree vertices and high-degree vertices separately.
The general theme of our algorithm is that low-degree vertices are
given priority to be in the MIS. They do not care about whether a
high-degree neighbor is in the MIS.
In other words, a \ld vertex adds itself to the MIS
if and only if it has no \ld neighbors in the MIS.
On the other hand, a \hd vertex removes itself
from the MIS whenever a \ld neighbor adds itself
to the MIS. Thus, \ld vertices and \hd vertices follow
different algorithms for adding themselves to
the MIS. We provide such algorithms
in detail in~\cref{sec:low-degree} and~\cref{sec:high-degree}.
However, with insertions and deletions of edges, $\mmax$
may change.
We provide a \emph{restart procedure} that allows vertices to reassign
their degree designations when $\mmax$ changes by enough.
The restart procedure for each vertex checks after every update incident to it
to determine
whether it is a \ld or a \hd vertex.
Finally, each vertex $v$ in the graph maintains a counter $c_v$ that
indicates the number of $v$'s \ld neighbors that are in the MIS.

More specifically, our algorithm contains the following procedures (described
briefly here and expanded upon in the following sections):

\begin{enumerate}
    \item Edge updates between two vertices, such that one of which is low-degree, are
        processed following the algorithms in~\cref{sec:low-degree}. There are
        two possible scenarios:
        \begin{enumerate}
            \item If an edge update causes a low-degree vertex $v$'s
                counter to become $0$, $c_v = 0$, then $v$ must add itself to
                the MIS. Then, it must inform its high-degree neighbors that it
                was added to the MIS. All high-degree neighbors are processed
                according to~\cref{alg:low-deg-enters}.
            \item If an edge insertion occurs between two low-degree vertices
                both in the MIS, then one of them, $v$, removes itself from the MIS
                and informs all neighbors. Then, $v$'s low-degree neighbors are
                processed using~\cref{alg:low-neighbor} and $v$'s high-degree
                neighbors are
                processed using~\cref{sec:low-deg-left}.
        \end{enumerate}
    \item Edge updates between two vertices, such that both of them are high-degree, are
        processed in~\cref{sec:high-degree}. There are two scenarios:
        \begin{enumerate}
            \item If any update causes a high-degree vertex to have no neighbors
                in the MIS, it adds itself to the MIS.
            \item If any such update causes any high-degree vertex $v$ to leave the
                MIS, then it must call~\cref{alg:high} to determine which of its
                neighbors should enter the MIS (with $v$ as the leader).
        \end{enumerate}
    \item A \emph{restart} procedure given by~\cref{alg:clean vertices}
        is called before handling high-degree vertices. This restart procedure
        ensures that the \emph{number of high-degree} vertices remains bounded
        by $\mmax^{1/3}$.~\label{item:high-restart}
    \item If a low-degree vertex $v$'s degree, $deg(v)$, exceeds some internally
        maintained
        threshold, $deg'(v)$, then $v$ changes its degree designation to
        high-degree. (Note that $v$ can become low-degree again
        via~\cref{item:high-restart}.)
\end{enumerate}

As defined above, a vertex $v$ is \hd if $deg(v) > deg'(v)$.\\
Our algorithm maintains the following invariants:

\begin{invariant}\label{inv:low}
    If a \hd vertex $u$ is in the MIS then none of its \ld neighbors
    $w \in N_{low}(u)$ have $c_w = 0$.
\end{invariant}

\begin{invariant}\label{inv:restart}
    Throughout the execution, $deg'(v) < 4m_{max}^{2/3}$ for all $v \in V$.
\end{invariant}

\cref{inv:low} ensures that a \hd vertex is in the MIS if and only if none of
its \ld vertices can enter the MIS. \cref{inv:restart} helps us maintain
that vertices designated as \ld have degree
$O\left(\mmax^{2/3}\right)$.

\subsection{Updates on Low-Degree Vertices}\label{sec:low-degree}

We first describe our algorithm on low-degree vertices.
Specifically, we describe what happens when an edge
insertion or deletion
occurs between two vertices where at least
one of these two vertices is a low-degree vertex.

Since the graph is initially empty (contains no edges), all
vertices are initially \ld. When a vertex which is \hd becomes
\ld and vice versa, it immediately notifies all its neighbors of
its new designation. When a low-degree vertex becomes high-degree,
it further needs to ensure that each of its \ld neighbors are in
the MIS if they can be. In other words, it needs to ensure that
none of its low-degree neighbors are depending on it to be in the MIS.
We describe this procedure in detail and
prove later on that this procedure is not too costly.

\iflong
\else
The full detailed procedure for low-degree vertices can be found in~\cref{app:low-deg}.
For the sake of  space, we describe the \emph{hard case} when an edge insertion occurs
between two low-degree vertices both of which are in the MIS. Then, one low-degree vertex
must remove itself from the MIS and inform its neighbors that it removed itself from the MIS.
Its low-degree neighbors must add themselves to the MIS if they have no low-degree neighbors in the MIS.
These low-degree neighbors perform the algorithm given in~\cref{alg:low-neighbor}.

\paragraph{Low-Degree Vertex Exits the MIS}
The vertex $v$ that removed itself from the MIS designates itself as the leader and coordinates a schedule
for its low-degree neighbors to enter the MIS one-by-one in some sequential order.
If a low-degree neighbor enters the MIS, it sends a message to all its neighbors that it entered the MIS.
Then, neighbors which occur later on in the schedule do not enter the MIS if an earlier neighbor entered.
We can afford to pay for the rounds of communication and messages in the amortized sense due to the
observation that \emph{at most two low-degree vertices leave the MIS after any edge update}.

We show how to handle high-degree neighbors of $v$ in the next section.
\fi

Upon an edge insertion, each vertex first sends its newest neighbor
an $O(1)$-bit message indicating whether it is in the MIS \emph{and}
whether it is \ld or \hd.

The algorithms we describe below follow this general intuition.

\paragraph{Edge Insertion between Two Low-Degree Vertices}
We denote by $(u,v)$ the edge that was inserted between two
initially \ld vertices.
If after this insertion w.l.o.g. $deg(u) > deg'(u)$,
then $u$ re-designates itself as a \hd vertex
and notifies all its neighbors.

Any vertex, w.l.o.g. $u$, that becomes \hd after the insertion performs
the following procedure.
If $u$ is in the MIS, each of $u$'s neighbors, $w \in N(u)$,
decreases its counter $c_w$ by one;
if any \ld neighbor
$w$ now has $c_w = 0$, then $w$ sends a message with $O(1)$
bits to $u$ informing $u$ that its $c_w$ is $0$. Let $W$
be the set of \ld neighbors $w$ of $u$ that have a $c_w=0$ and therefore want to add themselves
to the MIS. If $|W| > 0$, then $u$ removes itself from the MIS.
$u$ then determines an arbitrary sequential order for the vertices in $W$ to add
themselves to the MIS. $u$ sends each $w \in W$ a number using $O(\log n)$
bits indicating $w$'s order number.
Since the rounds are synchronous, each $w \in W$ adds itself
to the MIS sequentially in
this order if it does not have any \ld neighbors in the MIS.
If any $w \in W$ becomes dominated by another \ld neighbor before its
turn, it informs $u$ using an $O(1)$-bit message and $u$ updates
the turn numbers of all vertices after $w$ in the order. Pseudocode
for this algorithm is provided in~\cref{alg:low-neighbor}.

\begin{algorithm}
    \small
\caption{Low-degree neighbors entering MIS.}\label{alg:low-neighbor}
    \LinesNumbered
    \SetAlgoLined
    \KwResult{Neighbors $w \in N_{low}(u)$ add themselves to the MIS if none of
    their low-degree neighbors are in the MIS.}
    \textbf{Input: } Given $u$ which is a \ld vertex which removed
    itself from the MIS
    or which is a \ld vertex that is in the MIS and became \hd.\\
    $u$ sends $O(1)$-bit message to each of its neighbors to decrease their counter by $1$.\\
    \For{each $w \in N_{low}(u)$ (\ld neighbor of $u$)}
    {$w$ decrements its counter $c_w$ by 1.\\
        \If{$c_w = 0$}{$w$ sends $O(1)$-bit message to $u$
        indicating it wants to enter the MIS.}
    }
    Let $W$ be the set of \ld neighbors of $u$ that want to enter
    the MIS. $u$ determines an arbitrary sequential
    order for $w \in W$ to be added to the MIS.\\
    \For{each $w \in W$}{
            $u$ sends to $w$ a $O(\log n)$-bit message indicating its order.\\
            \If{no neighbor $x \in N_{low}(w)$ is in the MIS}{
                $w$ adds itself to the MIS.\\
                                                $w$ informs all neighbors it is in the MIS. \\
                Each neighbor $x \in N(w)$
                increments $c_{x}$ by $1$.\\
                High-degree neighbors of $w$ run
                the algorithm provided in~\cref{alg:low-deg-enters}.
        }
            \If{$w$ has a \ld neighbor which enters the MIS before its turn}{
                    $w$ informs $u$ using $O(1)$-bit message.\\
                    $u$ sends all $w' \in W$ after $w$ in their sequential
                    order a new turn number.\\
            }
    }
    \If{u left the MIS}{\For{each $w' \in N_{high}(u)$ (\hd neighbor of $u$)}{
        Perform the algorithm detailed in~\cref{sec:high-degree}.
    }}
\end{algorithm}

After~\cref{alg:low-neighbor} has finished running,
all of $u$'s \hd neighbors $w \in N_{high}(u)$ which do
not have any neighbors in the
MIS, perform the procedure detailed in~\cref{sec:high-degree}.

After updating the degree designation, the algorithm
continues assuming $(u,v)$ is an edge insertion
between a \ld vertex and a \hd vertex (described
below). If both $u$ and $v$ have $deg(u) > deg'(u)$ and $deg(v) > deg'(v)$,
and are both in the MIS,
then we perform~\cref{alg:low-neighbor} for both of them
and continue the algorithm as if $(u,v)$ is an edge insertion
between two \hd vertices
(described in~\cref{sec:updates-high-deg}).

If, however, $(u, v)$ is inserted between two low-degree
vertices and neither of which becomes \hd, we perform the following
algorithm depending on whether neither, both, or one of the
two vertices are in the MIS:
\begin{itemize}
    \item \textbf{At most one is in the MIS.} $u$ or $v$ change their counter
        $c_u$ and $c_v$, respectively, accordingly.
    \item \textbf{Both are in the MIS.}
        In this case, both $u$ and $v$ have informed
        each other that they are \ld and in the MIS.
        One of $u$ or $v$ (the one with the lower ID), w.l.o.g. assume $u$,
        removes itself from the MIS
        and sends to all its neighbors $N(u)$ an $O(1)$-bit
        message that it is no longer in the MIS.
        Then, it proceeds to follow the algorithm outlined
        in~\cref{alg:low-neighbor}.

We also need to take care of the high-degree neighbors
of $u$ after it removed itself from the MIS
and the \hd neighbors of the low-degree vertices $w \in N_{low}(u)$
which have added themselves to the MIS. We take care of these \hd
neighbors separately in~\cref{sec:high-degree}.
\end{itemize}

\paragraph{Edge Deletion between Two Low-Degree Vertices}
If neither $u$ or $v$ is in the MIS, nothing happens.
If w.l.o.g. $u$ is in the MIS and $v$ is not in the MIS,
$v$ decrements $c_v$ by $1$.
If $c_v = 0$ 
it adds
itself to the MIS and informs all neighbors.
Otherwise, it does nothing.
Again, if $v$ adds itself to the MIS, we have to take care of the high-degree
neighbors of $v$ which may be in the MIS; we describe this procedure
in~\cref{alg:low-deg-enters} (with $v$ as the leader).

\paragraph{Edge Insertion between One Low-Degree and One High-Degree Vertex}
Suppose $u$ is the low-degree vertex and $v$ is the high-degree vertex,
if $deg(u) > deg'(u)$ after the insertion, $u$ becomes \hd and
performs~\cref{alg:low-neighbor}. In this case, we treat the edge insertion
as an edge insertion between two \hd vertices and handle such an insertion
in~\cref{sec:updates-high-deg}.
Now suppose $u$ does not change its degree designation after the insertion.
If $u$ or $v$, or both, are not in the MIS, then, nothing needs to be done except for changing the counters.
If both $u$ and $v$ are in the MIS, then $v$ needs to perform the procedure
outlined in~\cref{alg:low-deg-enters}; this is the procedure
for a \hd vertex in the MIS when one of its \ld neighbors enters the MIS.
Obtaining this information
about $u$ requires $O(1)$ messages and $O(1)$ rounds of communication.

\paragraph{Edge Deletion between One Low-Degree and One High-Degree Vertex}
Suppose $u$ is the low-degree vertex and $v$ is the high-degree vertex.
First, if after the deletion $deg(v) \leq deg'(v)$, $v$ moves itself to $V_L$
and update all its neighbors, if $v$ is not in the MIS and has only \hd neighbors in the MIS, it must enter the MIS and all its \hd neighbors must perform the algorithm detailed in~\cref{sec:high-degree}; the update then continues as an edge deletion
between two \ld vertices.

Otherwise, if $u$ is in the MIS and $v$ has no more
low-degree neighbors in the MIS (i.e.\
$c_v = 0$ after the deletion),
it needs to perform the procedures outlined in~\cref{sec:low-deg-left}
when a \hd vertex is not in the MIS and no longer has any \ld neighbors in the MIS.
If this is not the case,
it does not need to do anything.
Note that it is never the case that $v$ was in the MIS
and $c_u = 0$ before the deletion by~\cref{inv:low}.

We now describe our algorithm for \hd vertices which
ensures that all high-degree vertices which have recent changes
in their low-degree neighbors add or remove themselves from the MIS
as necessary. We describe two related, but different procedures
that use as black boxes static, distributed MIS algorithms as subroutines.
One of the two procedures is conceptually simpler but obtains
worse dependence in the round and message complexity;
the other is conceptually more complex but obtains better
dependence.
By using our restart procedure (\cref{sec:restart})
for determining our degree threshold for each
vertex, a high-degree vertex can afford to inform all its high-degree
neighbors whenever it adds itself to the MIS. Thus, all high-degree
vertices know whether \emph{any} of its neighbors are in the MIS.
(This is in contrast to low-degree vertices which only know if their low-degree
neighbors are in the MIS). There are two cases high-degree vertices must handle
when dealing with updates that cause their low-degree neighbors to enter
or leave the MIS.

The detailed algorithm for handling all cases involving high-degree vertices is
given in~\cref{sec:high-degree}. We again describe the \emph{hardest case} when
several low-degree vertices enter the MIS and the high-degree neighbors must
leave the MIS.

\paragraph{Several High-Degree Neighbors Leave MIS} When one or more high-degree
neighbor must leave the MIS, we must find the set of high-degree vertices that
no longer have neighbors in the MIS and add them to the MIS. We perform this
procedure by waking up \emph{all} high-degree neighbors of high-degree vertices
that left the MIS. Then, we run some \emph{static, distributed MIS algorithm} on
the high-degree vertices that are awake.
As the static, distributed MIS algorithm is costly in terms of number of
messages, we must not call it on too many vertices.
To ensure that we run such an algorithm
on a small enough number of high-degree vertices, namely $O(m^{1/3})$ such
vertices, we call our restart procedure (detailed in the next section). During
our restart procedure, some high-degree vertices may become low-degree and we do
not run the static MIS algorithm on such vertices. We include the pseudocode for
our restart procedure in~\cref{alg:find-subgraph}. We define $v$ to be the
low-degree vertex which left the MIS, $L$ to be the set of low-degree
neighbors of $v$ which enter the MIS, $U'$ to be the set of high-degree
neighbors of each $w \in L$ which are in the MIS, and $U$ to be the set of
high-degree neighbors of each $u \in U'$. We use these definitions
in~\cref{alg:find-subgraph}.

By noting that our neighborhood of affected vertex has diameter $6$, we can use a more complicated procedure~\cref{alg:high} to obtain better message and round bounds by $O(\log^3n)$. 
We include this more complicated procedure in~\cref{sec:high-degree}.

\iflong
\subsection{Handling High-Degree Vertices}\label{sec:high-degree}

\subsubsection{Low-Degree Neighbor Enters the MIS}\label{alg:low-deg-enters}
We assume an edge insertion between a \ld vertex $v$
and a \hd vertex $u$ falls under this category if $v$ is in the MIS.
A low-degree vertex can also enter the MIS if it moves from $V_H$ to
$V_L$, or if one of its \ld neighbors left the MIS due to an update.

Suppose $u$ is a \hd vertex that satisfies~\cref{inv:low}
and is in the MIS.
If a low-degree neighbor $v \in N_{low}(u)$
enters the MIS, then $u$ must leave the MIS and inform all its high-degree
neighbors. We use one of two procedures to determine the
set of high-degree vertices that must enter the MIS
due to $u$ leaving the MIS. For convenience,
we denote the set of \hd vertices which want to enter
the MIS by $V_H'$.

There are several situations which may cause a \ld vertex to enter the MIS.
Furthermore, our algorithms described below rely on several characteristics of
the sets of vertices that are affected by the change. Specifically we define
these roles:

\begin{enumerate}
    \item \textbf{Leader $v$}: The leader $v$ coordinates the addition into the MIS of \hd vertices that want to enter the MIS, so that no collisions occur.
    \item \textbf{Set $U$}: The set of \hd vertices that want to \emph{enter}
        the MIS.
    \item \textbf{Set $U'$}: The set of \hd vertices that want to \emph{leave}
        the MIS.
    \item \textbf{Set $L$}: The set of \ld neighbors of $v$ that want to
        \emph{enter}
        the MIS after the current update.
\end{enumerate}

We set the above roles as follows for each of the below situations. We
consider the roles \emph{after} performing the restart procedure. This
means that no vertices switch from $V_L$ to $V_H$ or vice versa after being
assigned to the below roles.

\begin{enumerate}
    \item \textbf{Edge insertion $(u, v)$ where $u$ is \ld and $v$ is \hd and
        $u, v$ are in the MIS}: $u$ is the leader, $U' =
        \left\{v\right\}$, $U$ is the set of \hd neighbors in $N_{high}(v)$
        which have no neighbors in the MIS, $L = \emptyset$.
    \item \textbf{$v$ leaves $V_H$ and moves to $V_L$, then enters the MIS}:
        $v$ is the leader, $U'$ is the set of \hd neighbors in
        $N_{high}(v)$ which are in the MIS, $U$ is the set of \hd neighbors
        $N_{high}(w)$ of vertices $w \in U'$ which have no neighbors in the MIS,
        $L = \emptyset$.
    \item \textbf{Edge insertion $(u, v)$ where both $u$ and $v$ are \ld and $u,
        v$ are in the MIS}: w.l.o.g. $v$ leaves the MIS. $v$ is the leader,
        $U$ consists of the \hd neighbors in $N_{high}(v)$ and
        the set of \hd neighbors
        $N_{high}(w)$ of vertices $w \in U'$ that have no
        neighbors in the MIS, $U'$ consists of the \hd neighbors that are
        in the MIS, $N_{high}(w)$, of $w \in N_{low}(v)$ where $w$ was added to
        the MIS, and
        $L$ consists of the set of \ld neighbors $w \in N_{low}(v)$ whose $c_w =
        0$.
    \item \textbf{$v$ leaves $V_L$ and moves to $V_H$, $v$ was originally in the
        MIS}: Suppose $v$ leaves the MIS when it moves into $V_H$. $v$ is the
        leader, $U$ consists of $w \in N_{high}(v)$ and
        the set of \hd neighbors
        $N_{high}(w)$ of vertices $w \in U'$
        that have no neighbors in
        the MIS, $L$ consists of $w \in N_{low}(v)$ where $c_w = 0$, and $U'$
        consists of neighbors of $w$ in $N_{high}(w)$ that are in the MIS
        for each $w \in N_{low}(v)$ that entered the MIS.
    \item \textbf{Edge deletion between $(u, v)$ where $u, v$ are \ld and $u$
        is in the MIS, $c_v = 1$}: $v$ enters the MIS in this case.
        $v$ is the leader, $U'$ consists of $w \in
        N_{high}(v)$ where $w$ is in the MIS,
        $U$ is
        the set of \hd neighbors
        $N_{high}(w)$ of vertices $w \in U'$
        that have no neighbors in the MIS,
        $L = \emptyset$.
\end{enumerate}

\paragraph{Finding an MIS in the Subgraph Induced by $V_H'$}
Suppose we are given some subset of \hd vertices $V_H' \subseteq V_H$.
Our first procedure uses the result of~\cite{GGR20} as a black box on
the subgraph induced by $V_H'$. To use the result of~\cite{GGR20},
we must first determine $V_H'$ such that all vertices in $V_H'$
know to participate in running the static MIS procedure
to determine the set of vertices from $V_H'$ which would enter the MIS.
We first perform the following procedure provided
in~\cref{alg:find-subgraph}
so that all vertices in $V_H'$ receive knowledge
of their participation in the MIS algorithm. Then,
each vertex in $V_H'$ runs~\cite{GGR20}, given in~\cref{thm:mis-subgraph},
as a black box to obtain the set of vertices which need to enter
the MIS.
Our procedures also
uses a \emph{restart} algorithm
which we describe in~\cref{alg:clean vertices}
of~\cref{sec:restart}.
\fi

\begin{algorithm}\caption{Find MIS within $V_H'$}\label{alg:find-subgraph}
\footnotesize
    \SetAlgoLined
 	\LinesNumbered
    \KwResult{A set of \hd vertices which enter the MIS.}
    \textbf{Input: } Leader $v$, sets $U$, $U'$ and $L$ as defined above.\\
    Suppose vertex $v$ is the designated leader of this procedure.\\
    $v$ informs all neighbors it entered/left the MIS using $O(1)$-bit messages.\\
    \If{$v$ entered the MIS}{
    Let $B$ be the set of all \hd neighbors of $v$.}
    \Else{\Comment{If $v$ exited the MIS}\\
        Recall $L$ is the set of all \ld neighbors $w$ of $v$ that want to enter the
        MIS because their $c_w = 0$.\\
        Perform~\cref{alg:low-neighbor} on $L$ to determine vertices that enter the MIS.\\
        Let $J$ be the set of vertices of $L$ picked to enter the MIS.\\
        Let $B$ be the set of all \hd neighbors of $v$.\\
        \For{each $a \in J$}{
                        Every high-degree neighbor of $a$
            becomes part of the set $B$.
        }
    }
    Let $W$ be the set of all \hd neighbors of $u \in B$.\\
    Perform~\cref{alg:clean vertices} (the restart procedure) on $\{v\} \cup J
    \cup B \cup W$.\\
    \For{each $u \in U'$}{
        $u$ leaves the MIS.\\
        $u$ informs all its
        high-degree neighbors
        that it left the MIS using $O(1)$-bit messages.\\
            }
    \For{each $w \in J$}{
                $w$ enters the MIS.\\
                $w$ informs all its neighbors that it entered the MIS using
                $O(1)$-bit messages.\\
                \For{each $x \in N(w)$}{
                    $x$ increments $c_x = c_x + 1$.\\
        }
    }
    \For{each $u \in U$}{
        $u$ informs all its high-degree neighbors that it is
        part of $V_H'$.\\
        All vertices which are in $V_H'$ know which of its neighbors
        are in $V_H'$\\
        Each vertex $w \in V_H'$ runs the algorithm given
        by~\cref{thm:mis-subgraph}
        to find an MIS in the induced subgraph defined by $V_H'$.\\
    }
                                            \end{algorithm}
\iflong
\begin{theorem}[Corollary 1.2~\cite{GGR20}]\label{thm:mis-subgraph}
    There exists a deterministic distributed algorithm that computes a maximal independent
    set in $O(\log^5 n)$ rounds in the \congest model.
\end{theorem}

Since the algorithm of~\cite{GGR20} operates in the distributed model where
each vertex initially does not know the topology of the graph (except
for its set of adjacent neighbors), we can
directly apply this algorithm to our induced subgraph $V_H'$ as
initially all vertices in $V_H'$ do not know the entire topology
of $V_H'$ but do know their set of neighbors (each of which has
a unique ID).

\paragraph{Finding an MIS from a Partial MIS}
Our second approach relies on an input-respecting
subroutine that obtains an MIS from a partial MIS; although
this approach is technically more complex,
it obtains a round complexity that is better
than the one obtained via our first approach.
We use the following algorithm which incorporates
the deterministic algorithm of~\cite{CPS20}
to accomplish this goal. We prove
an extension of~\cite{CPS20} in~\cref{lem:cps}
to handle the case when the input is a
graph with some number of vertices which are
already in the MIS
and use this procedure as a black box in
our detailed procedure provided
in~\cref{alg:high}.

\begin{algorithm}
\caption{High-Degree Vertices Entering MIS.}\label{alg:high}
    \LinesNumbered
    \SetAlgoLined
    \mysize
    \selectfont
    \KwResult{A set of \hd vertices which enter the MIS.}
    Suppose vertex $v$ is the designated leader of this procedure,
    $v$ informs all neighbors it entered/left the MIS using $O(1)$-bit messages.\\
    \If{$v$ entered the MIS}{
    Let $B$ be the set of all \hd neighbors of $v$.\\}
    \Else{
        Recall $L$ is the set of all \ld neighbors $w$ of $v$
                        with $c_w = 0$.\\
        Perform~\cref{alg:low-neighbor} on $L$ to determine vertices that enter the MIS.\\
        Let $J$ be the set of vertices of $L$ picked to enter the MIS.\\
                \For{each $a \in J$}{
                        Every high-degree neighbor of $a$
            becomes part of the set $B$.
        }
    }
    Let $W$ be the set of all \hd neighbors of $u \in B$.\\
    Perform~\cref{alg:clean vertices} (the restart procedure)
    on $\{v\} \cup J \cup B \cup W$.\\
                            \For{each $w \in J$}{
            $w$ enters the MIS.\\
            $w$ informs all its neighbors that it entered the MIS using
            $O(1)$-bit messages.\\
            \For{each $x \in N(w)$}{
                    $x$ increments $c_x = c_x + 1$.\\
            }
    }
    \For{each $u \in U'$}{
        $u$ leaves the MIS.\\
        $u$ informs all
                        high-degree neighbors
        that it left the MIS using $O(1)$-bit messages.\\
            }
    \For{each $u \in U$}{
        $u$ informs all \hd neighbors (that remained \hd)
        in $U'$ it's part of $V_H'$.\\
        \If{$u$ is a neighbor of $v$}{
                $u$ also informs $v$ that it is a part of $V_H'$.
        }
    }
    \For{each $w \in U'$}{
        \If{received a message from a vertex $u \in U$}{
            $w$ picks one neighbor $g \in N_{low}(w) \cap(J \cup \{v\})$
            and informs $g$ it has
            \hd neighbors in $V_H'$.
        }
    }

    \For{each $y \in J$ that received a message from a vertex $u \in U$}{
        $y$ informs $v$ of the existence of $V_H'$ (but not the members of
        $V_H'$).
    }
                                 $v$ becomes the leader in our procedure.\\
         $v$ coordinates computing a MIS among
            $w \in V_H'$
                        using~\cref{lem:cps} on subgraph induced by $\{v\} \cup U' \cup
            U \cup J$.\\
                                                                                                                                 If any $w\in V_H'$ enters the MIS, it informs all its high-degree neighbors.\\
         This procedure proceeds until no additional vertex $w \in V_H'$ wants to enter the
         MIS.\\

\end{algorithm}

We provide a self-contained description of the algorithm of~\cite{CPS20} in~\cref{app:cps}.

\begin{lemma}[Modified Theorem 1.5~\cite{CPS20}]\label{lem:cps}
    Given a connected graph $G = (V, E)$,
    a leader $v$ which is distance $D$ away from all vertices in the graph
    and a set of vertices $S \subseteq V$ already in the MIS,
    there exists a deterministic MIS algorithm that finds
    an MIS among all vertices $V \setminus S$ in $G$
    that completes in $O(D\log^2 |V|)$ rounds
    and sends at most $O(D|E|\log^2 |V|)$ messages
    in the $\textsc{Congest}$ model.
\end{lemma}

\begin{proof}
    We use the same algorithm as that presented in Theorem 1.5 of~\cite{CPS20}.
    Since $G$ is connected, the paths from all vertices in $V$ to $v$ (and vice
    versa) form a tree with $v$ as the root.
    All vertices $s \in S$ and all neighbors $N(s)$ of $s$
    do not participate in the conditional expectation calculation of the
    algorithm described in Section 4.1 of~\cite{CPS20}. All other
    vertices perform the calculations associated with the conditional expectation
    method. Then, as in
    Theorem 1.5 of~\cite{CPS20}, all conditional expectation calculations
    are aggregated at $v$ by a standard aggregation algorithm:
    each intermediate vertex
    sums the conditional expectation values of all its children
    and send this value to its parent.
    The leader $v$ takes the aggregate value, determines a bit,
    and then sends each bit of the seed to the entire graph.
    The rest of the algorithm proceeds as in~\cite{CPS20} in $O(D\log^2 |V|)$
    rounds. During each round, each vertex sends at most one message to each of
    its neighbors of $O(\log |V|)$ bits. Hence, the total number of messages
    sent is $O(D|E|\log^2 |V|)$.
\end{proof}

\subsubsection{Low-Degree Neighbor Leaves MIS}\label{sec:low-deg-left}
We assume an edge deletion between a \ld vertex $u$
and a \hd vertex $v$ falls under this category if $u$ is in the MIS.
If a low-degree neighbor $u$ leaves the MIS and its high-degree neighbor $v$ does not have
any neighbors in the MIS, $v$ enters the MIS and informs all its high-degree neighbors.
All \hd neighbors $N_{high}(v) \cup \{v\}$ perform the restart procedure provided
in~\cref{alg:clean vertices}.

\subsubsection{Edge Updates Between High-Degree
Vertices}\label{sec:updates-high-deg}

In this section, we describe our algorithms for edge updates between
two \hd vertices.

\paragraph{Edge Insertion Between Two High-Degree Neighbors}
If at most one of the endpoints is in the MIS, both do nothing except for changing their counter.
If both high-degree neighbors are in the MIS, one leaves (e.g.\ $u$) and
performs~\cref{alg:find-subgraph} or~\cref{alg:high} with $u$ as the leader.

\paragraph{Edge Deletion Between Two High-Degree Neighbors}
If w.l.o.g. the degree of one of the vertices (say $v$),
$deg(v)$, becomes less than $deg'(v)$, then $v$
becomes \ld and informs all its neighbors. If $v$ has
only \hd neighbors which are in the MIS, $v$ enters the MIS and
all its \hd neighbors perform~\cref{alg:find-subgraph} or~\cref{alg:high}
with $v$ as the leader.
After that, the update continues as an edge deletion between
one \hd vertex and one \ld vertex.
If the degrees of both of the vertices $u$ and $v$, $deg(u)$ and $deg(v)$,
become less than $deg'(v)$ and $deg'(u)$, respectively,
then both of them become \ld,
inform all their neighbors and enter the MIS if needed;
the update continues as edge deletion between two \ld neighbors.

Otherwise, if no degree assignment changes, and neither
of the endpoints is in the MIS, nothing happens.
If exactly one e.g.\ $u$ is in the MIS, the other, e.g.\ $v$,
adds itself to the MIS if it has no neighbors in the MIS
and informs all its neighbors.~\cref{alg:clean vertices} is called
on all neighbors $N_{high}(v) \cup \{v\}$.
\fi

%% file: m-max-restart.tex
\subsection{Restart Procedures}\label{sec:restart}

The first of two restart procedures called
whenever~\cref{alg:find-subgraph} or~\cref{alg:high} runs or when a \hd
vertex informs its neighbors it entered
or left the MIS. This procedure ensures that
all \hd neighbors of such vertices become \ld
if necessary. This might occur
if the number of edges in the graph increases over time.
We first provide intuition for our algorithm and then
a detailed description of it.

\subsubsection{Maintaining Degree Bounds under $m_{max}$}
\paragraph{Intuition} Since the vertices in the graph do
not know the number of the edges in the graph, the partition
to \hd vertices and \ld vertices can be meaningless.
Therefore, we guess for each vertex
the maximum number of edges to ever exist
in the graph, which we will tighten
during the process. We present a restart procedure that maintains
\cref{inv:restart} and proves~\cref{thm:restart}
with $O(1)$ amortized number of
rounds and $O(\mmax^{2/3})$ amortized communication
complexity, where $\mmax$ is the maximum number of edges
that in the graph throughout the updates.
For the remainder of this section, we let $m$ denote $\mmax$.

\paragraph{The Algorithm} Let $\deg'(v)$ be the \emph{approximate movement degree}; when a vertex $v$ has degree
$\deg(v) \ge \deg'(v)$ it considers itself a \hd vertex and if $deg(v) < \deg'(v)$
it considers itself a \ld vertex. Whenever its degree crosses this threshold
it informs all its neighbors of the change.
Since the algorithm
starts from an empty graph (no edges),
for each $v \in V$, $\deg(v) = 0$ and we initialize $\deg'(v) = 2$.

It might be the case that for some vertices $\deg'(v) << m^{2/3}$
and therefore more than $m^{1/3}$ vertices will be in
$V_H$. To be able to perform~\cref{alg:find-subgraph}
or~\cref{alg:high} and allow
\hd vertices to inform their \hd neighbors of whether they
entered or left the MIS,
we must first ``clean'' $V_H$ and
move the unnecessary vertices to $V_L$. This procedure
is described below
in~\cref{alg:clean vertices}, which uses ~\cref{alg:bfs}, and~\cref{alg:approx m} as subroutines.
Namely, the crux of the subroutines is to determine the total
degree of all the vertices in the subgraph and move vertices
which have degrees that are too small into $V_L$.
Denote by $G' = (V', E')$ the subgraph induced by the vertices that
participate in~\cref{alg:find-subgraph} or~\cref{alg:high}.
Let this subgraph consist of vertices $\{v\} \cup L \cup B \cup W$ such that $v$ is the vertex
that started the restart (the leader).\\
We call~\cref{alg:bfs} on $G'$.

\begin{algorithm}
    \small
    \caption{Construct BFS tree.}\label{alg:bfs}
    \LinesNumbered
    \SetAlgoLined
    \KwResult{A BFS tree $T'$ in a carefully chosen subgraph $G'$ of $G$}
    \For{each $u \in N(v) \cap V'$ (neighbor of $v$ in $G'$)}
    {$v$ sends $O(1)$-bit message to $u$ indicating that $v$ is $u$'s
    parent in $T'$.\\
    }
    \While{there exists at least one \hd
    vertex in $V'$ that does not yet have a parent}{
    \For{each $u$ that received a message from its parent in the previous round}{
    $u$ sends $O(1)$-bit message to $w \in N(u)\cap V'$ if $w$ does not have a parent yet,
    indicating that $u$ is $w$'s parent in $T'$.\\
    If $u$ receives multiple parent messages, it arbitrarily chooses one to be
    its parent.\\
    $w$ stores $u$ as its parent.\\
    $w$ informs all its neighbors in $G'$ it has a parent.\\
    }}
    The process ends when all the vertices in $G'$ have a parent (except $v$
    which is the root).
\end{algorithm}

\begin{algorithm}
    \small
    \caption{Estimate a lower bound of $m$ in $G'$.}\label{alg:approx m}
    \LinesNumbered
    \SetAlgoLined
    \KwResult{A lower bound estimation of $m$.}
    Construct a BFS tree $T'$ according to~\cref{alg:bfs} and mark $v$ as the root.\\
    \For{each $u$ which is a child of $v$ in $T'$}
    {$v$ sends $O(1)$-bit message to $u$ asking what is the sum of the degrees
        of all vertices in $u$'s subtree.\\}
    \While{exists a vertex $u$ that received a message from its parent in the
    previous round}{
        \If{$u$ has children in $T'$}{$u$ sends the same message to all
        its children and waits for a response.\\
        $u$ sums the values it received from its children and adds $\deg(u)$ to the sum.\\
        $u$ sends the sum to its parent.\\}
        \If{$u$ has no children in $T'$}{$u$ sends $\deg(u)$ to its parent.}
    }
    \If{$v$ receives a non-zero sum of the degrees of all its children}
    {$v$ computes a final sum of the values received from its children
        and adds $\deg(v)$ to the sum. Let this sum be $S$.\\
    \For{each child $u$ of $v$ in $T'$}
    {$v$ sends a $O(\log n)$-bit message to $u$ containing $S$.}
    \For{each $u$ that received $S$ from its parent}{$u$ sends $S$ to its children.\\}}
    The process ends when all the vertices in $G'$ know $S$.
\end{algorithm}

In~\cref{alg:approx m},
since each vertex has only one parent (because $T'$ is a BFS tree), the
degree of each vertex is added to the sum only once, and therefore the sum
is given by $S = \sum_{v \in V'} \deg(v)$.

\begin{algorithm}
    \caption{Restart procedure.}\label{alg:clean vertices}
    \LinesNumbered
    \SetAlgoLined
    \KwResult{Some vertices $v$ with degree $\deg(v) <4m^{2/3}$ move to $V_L$
        and the number of vertices in remaining $G'$ (to be used in
        \cref{alg:find-subgraph} or~\cref{alg:high})
        is at most $2m^{1/3}$.}
        Let $S$ be the sum calculated in~\cref{alg:approx m}.\\
        Assume as in the previous algorithms $v$ is the leader and
        root of the tree.\\
        //$v$ broadcasts using $T'$ and tells all its descendants to become \ld (i.e. move to $V_L$) if necessary.\\
        \For{each $u$ which is a child of $v$ in $T'$}
        {$v$ sends $O(\log(n))$-bit message to $u$ telling it to become \ld
        (i.e. move to $V_L$) if $\deg(u)<S^{2/3}$.\\}
        \For{a descendant $u$ of $v$ that got the message}{$u$ propagates the
            message to its children in $T'$ (if any).\\ \If{$\deg(u) < S^{2/3}$}{$u$ considers itself as a \ld vertex (i.e. moves to $V_L$).\\
        If $u$ has no \ld neighbors in the MIS it enters the MIS.\\
        $u$ notifies all its neighbors in $G$ that it is \ld and if it is in the MIS or not.\\
         $u$ updates $\deg'(u) = 2\deg(u)$.\\}}
\end{algorithm}

Note that after performing~\cref{alg:clean vertices},
a vertex that had only \hd neighbors in the MIS might move to $V_L$
and enter the MIS. After that, all its \hd neighbors must leave
the MIS and perform another restart (by calling~\cref{alg:find-subgraph}
or~\cref{alg:high}).
We show in our analysis that this ``chain'' of restarts is constant on average.

\paragraph{Analysis of the Restart Procedure for $m_{max}$}
We provide the analysis of the restart procedure
described above here. This analysis will be crucial
in our analysis of our main algorithm provided in~\cref{sec:analysis-max}.
In the following analysis, we assume that $m = m_{max}$.

\begin{observation}\label{obs:S}
Throughout the execution of~\cref{alg:clean vertices}, a vertex $u$ in $G'$
moves from $V_H$ to $V_L$ if and only if its degree $\deg(u)$ is smaller than $S^{2/3}$.
\end{observation}

\begin{lemma}\label{thm:restart}
    After performing \cref{alg:clean vertices}, the number of
    vertices that remain in $G'$ is at most $2m^{1/3}$.
\end{lemma}
\begin{proof}
    By~\cref{obs:S}, each vertex $u$ in $G'$ with degree $\deg(u) < S^{2/3}$ moves to $V_L$. The vertices that remain in $G'$ are the vertices that have degree greater or equal to $S^{2/3}$.
    Since $S$ is the sum of the degrees in $G'$ before~\cref{alg:clean vertices} run, ~\cref{obs:S} implies that the number of vertices that remain in $G'$ is at most $S^{1/3}$ (otherwise, the sum of the degrees in $G'$ would be greater than $S^{1/3}\cdot S^{2/3}=S$).
Since $S = \sum_{v \in V'} \deg(v) \le 2m$,
we have $S^{1/3} < 2m^{1/3}$ vertices, which means that the number of
    vertices that remain in $G'$ is at most $2m^{1/3}$
\end{proof}

\begin{lemma}\label{lemma1}
    If a vertex $v \in V'$ moves from $V_H$ to $V_L$ during~\cref{alg:clean
    vertices}, then $\deg(v) < 2m^{2/3}$.
\end{lemma}

\begin{proof}
By~\cref{obs:S}, a vertex $v \in V'$ that moves from $V_H$ to $V_L$ during~\cref{alg:clean
    vertices}, moves only because its degree $\deg(v) < S^{2/3}$.
    Since $S = \sum_{v \in V'} \deg(v) \le 2m$, we get that $\deg(v) < S^{2/3} < 2m^{2/3}$.
\end{proof}

We can now prove~\cref{inv:restart}.

\begin{lemma}\label{lem:d'v} Throughout all the updates, $\deg'(v) < 4m^{2/3}$ for all $v \in V$
    when $m \geq 1$.
\end{lemma}
\begin{proof}
Assume for contradiction that an update exists
such that after it occurs there is a vertex $v$
with $\deg'(v) \ge 4m^{2/3}$.
Denote by $t$ the update in which $\deg'(v)$ was last updated.
$\deg'(v)$ can only change when $v$ moves from $V_H$ to $V_L$.
Let $\deg_t(v)$ be the degree of $v$ at update $t$.
Since $v$ moved from $V_H$ to $V_L$ at update $t$, it changed $\deg'(v)$ to be $\deg'(v) = 2\deg_t(v)$
By \cref{lemma1},
we know that if $v$ moves from $V_H$ to $V_L$ then
its degree $\deg(v) < 2m^{2/3}$, in particular, at update $t$, $\deg_t(v) <
2m^{2/3}$. By our assumption, $\deg'(v) \geq 4m^{2/3}$, which implies that
$\deg'(v) = 2\deg_t(v) \geq 4m^{2/3}$ and $\deg_t(v) \geq 2m^{2/3}$, a contradiction.
\end{proof}

\subsubsection{Runtime Analysis of the Restart Procedure for $\mmax$}\label{subsec:restart-analysis}
We first provide a worst case analysis and afterwards provide an amortized analysis.
In \cref{alg:bfs}, since the diameter of $G'$ is at most $6$,
building $T'$ takes $O(diameter) = O(1)$ rounds. In each round each vertex $v$
sends a constant number of messages to all of its neighbors in $G'$, so the total number of messages would be
$O(|E'|)$ messages. The same analysis as~\cref{alg:bfs} holds for~\cref{alg:approx m}.
As we described in each of these procedures, the messages have size $O(\log n)$.

In~\cref{alg:clean vertices}, the broadcast costs $O(diameter) = O(1)$ rounds and $O(|E'|)$ messages as explained for~\cref{alg:bfs}. However, each vertex $u$ that moves from $V_H$ to $V_L$ during this algorithm, must check if it needs to enter the MIS and notify all its neighbors about this change.

Let $V_{low}$ be all the vertices from $V'$
that moved from $V_H$ to $V_L$ during the restart, and let $S' = \sum_{v\in
V_{low}} \deg(v)$ be the sum of
their degrees in $G$.
All the vertices that move to $V_L$ need to notify their neighbors about the movement and might enter the MIS, however, this should be done sequentially.
Therefore, notifying all the neighbors about the change costs at most $O(|V_{low}|)$ rounds and $O(S')$ messages.
Overall we pay for one restart $O(|V_{low}|)$ rounds and $O(|E'| + S')$ messages.

\paragraph{Amortized analysis}
We obtain in the worst case that one restart costs $O(|V_{low}|)$ rounds and $O(|E'| + S')$ messages.
We present now the amortized analysis of one restart.
\begin{lemma}\label{lem:m-max-runtime}
    \cref{alg:bfs},~\cref{alg:approx m}, and~\cref{alg:clean vertices}
    take $O(1)$ amortized rounds and $O(m^{2/3})$ amortized
    messages over the set of updates in our graph.
\end{lemma}

\begin{proof}
    As defined above, $V_{low}$ is the set of all the vertices from $V'$
    that moved from $V_H$ to $V_L$ during \cref{alg:clean vertices}
    and $S'$ is the sum of their degrees in $G$ and $L$
    is the set of \ld vertices that enter the MIS.
    $|E'| = |E(V' \setminus L)| + \sum_{a \in L}\deg(a)$. By construction of $L$, all
    \ld vertices in $L$ are those that \emph{enter the MIS}. Thus, we can charge
    the cost of $\sum_{a \in L} \deg(a)$ to the cost of adding $L$ to the MIS. This
    cost is computed later in~\cref{sec:analysis-max}
    by~\cref{lem:insertion-low}.

    Then, $|E(V' \setminus L)| \le m^{2/3} + \sum_{v\in V_{low}} \deg(v) = m^{2/3} + S'$.
    This inequality holds due to the following argument.
    Recall that $G'$ consists of the vertex sets $\{v\} \cup L \cup B \cup W$
    and all edges in the induced subgraph containing these vertices. If $v$ is
    \ld, then $\deg(v) = O(m^{2/3})$; if $v$ is \hd, then it only communicates
    with other \hd vertices, so its degree in $G'$ is $< m^{1/3} + |\{(v, w) | (v, w) \in E \wedge w \in
    V_{low}\}|$. Since
    $B$ and $W$ only contain \hd vertices (including those from $V_{low}$),
    the number of vertices $|B \cup W| - |V_{low}| < 2m^{1/3}$ according to
    \cref{alg:clean vertices}, and the number of induced edges between vertices
    in $B \cup W$ is $< 4m^{2/3} +  \sum_{v\in V_{low}} \deg(v)\leq 4m^{2/3} + S'$.
    So the total number of messages
    $O(|E(V' \setminus L)| + S') = O(m^{2/3} +S')$. Therefore, we only
    need to account for the extra cost of $O(\sum_{v\in V_{low}} \deg(v))
    = O(S')$ messages and $O(|V_{low}|)$ rounds.
    We can amortize this cost via the following charging
    argument. Each $v\in V_{low}$ changes $deg'(v)$ to be $2\deg(v)$
    when the transition occurs, so it would take at least $O(\deg(v))$
    updates \emph{that are adjacent to $v$} until $v$
    moves back to $V_H$. Since each update is adjacent to at most two vertices,
    in each such update before
    $v$ moves back to $V_H$, we
    can afford to pay $O(1)$ coins for the number of rounds and messages
    incurred during the last restart process $v$ participated in.
    If $v$ never returns to $V_H$,
    we can use the fact that this is its last restart
    and pay for it in advance by paying another extra
    $O(1)$ coins
    on each edge insertion adjacent to $v$ before the
    \emph{previous} restart.
    To conclude, we need to
    pay $O(1)$ amortized rounds and $O(m^{2/3})$ amortized messages on the restart process.
\end{proof}

We now analyze how we can pay for a chain of restarts.
\iflong
\begin{lemma}\label{lem:const}
During each update, the number of \ld vertices that leave the MIS is constant.
\end{lemma}

\begin{proof}
    A vertex $v \in V_L$ can leave the MIS in two cases:
    \begin{itemize}
        \item $v$ is in the MIS, and after the update $\deg'(v) < \deg(v)$,
            $v$ moves to $V_H$.
        \item both $u \in V_L$ and $v$ are in the MIS before the update,
            and the update is the insertion of $(u,v)$.
    \end{itemize}
    In both cases at most $2$ vertices move to $V_H$ and might leave the MIS. If the vertices stay in $V_L$, then at most $1$ leave the MIS.
    Since those are the only relevant cases, we get that the number of \ld vertices that leave the MIS is constant.
\end{proof}
\fi

\begin{restatable}{lemma}{mmaxrestartcost}
All restarts that happen during an update costs $O(1)$
amortized rounds and $O(m^{2/3})$ amortized messages.
\end{restatable}

\iflong
\begin{proof}
    After the first restart, only a vertex the moved from $V_H$ to $V_L$ and had only
    \hd neighbors in the MIS can cause the occurrence of another restart.
    By~\cref{lem:const}, the number of \ld vertices that leave the MIS in an
    update is constant.

    Since the algorithm starts from an empty graph
    and all vertices are initially in the MIS, when the vertex leaves the MIS
    we can pay another $O(1)$ rounds and $O(m^{2/3})$ messages for the restart procedure it might cause.
    Note that the rest of the restart's cost (which is $O(S')$ messages and
    $O(V_{low})$ rounds) is payed as we described above in~\cref{lem:const}
    for one restart.
    This works since in each such restart in the chain,
    only \ld vertices that
    \emph{did not participate in a previous restart} for this
    update participate. Formally, if given a chain of restarts
        $r_1,r_2,..,r_k$ and the corresponding sets of vertices that move from $V_H$
    to $V_L$, $V^1_{low}, V^2_{low}, ...,
    V^k_{low}$,
        then for any pair of $i, j$ where
    $1 \le i < j \le k, V^i_{low}\cap V^j_{low} = \emptyset$.
    Therefore, the payment for each such restart is given by different vertices.

    To conclude, we get that the cost of a chain of restarts is
    $O(1)$ amortized rounds and $O(m^{2/3})$ amortized messages.
\end{proof}
\fi

%% file: main-analysis.tex
\subsubsection{Analysis of the Main Algorithm under
$m_{max}$}\label{sec:analysis-max}

\iflong
In this section, we provide the analysis of our main algorithm
which uses the analysis of our restart procedures detailed above.
\else
    We provide the full analysis of our main algorithm under $m_{max}$ in~\cref{app:max-analysis}. The analysis of our main
    algorithm follows straightforwardly from~\cref{inv:low}
    and~\cref{inv:restart} and the analysis of our restart procedures above. We
    obtain the following theorem with respect to $m_{max}$.
\fi
\iflong
\begin{lemma}\label{lem:deletion-low}
    Given an edge deletion between two low-degree vertices $(u, v)$
    where w.l.o.g. $u$ is in the MIS and $c_v = 0$ after the deletion,
    $v$ adds itself to the MIS
    after $O(1)$ rounds of communication
    and after sending $O(m_{max}^{2/3})$ messages.
\end{lemma}

\begin{proof}
    By~\cref{inv:restart}, $deg'(v) < 4m_{max}^{2/3}$ for vertex $v$.
    Since $v$ is low-degree, $deg(v) \leq deg'(v)$. Thus, if $c_v = 0$,
    $v$ adds itself to the MIS and informs all neighbors.
    Since $v$ has at most $deg(v)$ neighbors, this requires
    at most $O(m_{max}^{2/3})$ messages which are
    $O(1)$-bit each and $O(1)$ rounds of
    communication.
\end{proof}

\begin{lemma}\label{lem:insertion-low}
    Given an edge insertion between low-degree vertices $(u, v)$
    where both are in the MIS,~\cref{alg:low-neighbor} requires $O(1)$
    amortized rounds and $O(m_{max}^{2/3})$ amortized messages using
    $O(\log n)$-bit messages.
\end{lemma}

\begin{proof}
    The correctness of the procedure follows from the observation
    that only vertices which do not have low-degree
    neighbors in the MIS are added
    to the MIS.
    We prove the amortized message complexity using a similar charging
    argument to that given in~\cite{AOSS18}. By ~\cref{lem:const}, the number of \ld vertices that leave the MIS in an update is constant.
    In addition, by ~\cref{inv:restart}, every time a \ld vertex enters the MIS, it needs at most $O(1)$ rounds and $O(m_{max}^{2/3})$ messages to notify all its neighbors. Since the algorithm starts with an empty graph and all the vertices are in the MIS, each time a vertex leaves the MIS it pays another $O(1)$ rounds and $O(m_{max}^{2/3})$ messages to its next insertion to the MIS.
    Therefore, the total cost of~\cref{alg:low-neighbor} is $O(1)$
    amortized rounds and $O(m_{max}^{2/3})$ amortized messages.
    Note that the insertion of \ld vertices to the MIS could cause
    \cref{alg:find-subgraph} or~\cref{alg:high} to run.
    Their costs are analyzed below in~\cref{lem:high-deg-computation}
    and~\cref{cor:high-2-runtime}.
                                                                            \end{proof}

We move onto the runtimes of edge updates between a \hd vertex and a \ld vertex.

\begin{lemma}\label{lem:high-deg-computation}
    \cref{alg:high} requires $O(\log^2 n)$ amortized rounds and
    $O(\mmax^{2/3}\log^2 n)$ amortized messages.
\end{lemma}

\begin{proof}
    The first set of steps of the procedure builds the graph $G' = (V', E')$
    which undergoes the restart procedure given by~\cref{alg:clean vertices}.
    We can charge building this graph to the cost of the restart procedure.
    Then, by~\cref{lem:m-max-runtime}, the cost of performing the restart is
    $O(1)$ amortized rounds and $O(\mmax^{2/3})$
    amortized messages.

    The next set of steps adds \ld vertices in $L$ to the MIS.
    By~\cref{lem:insertion-low}, these steps require $O(1)$ amortized rounds
    and $O(m^{2/3})$ amortized messages.
        Then, vertices in $U'$ inform their \hd neighbors that they are leaving the
    MIS.
    By~\cref{thm:restart}, the number of vertices that
    remain in $G'$ after the restart is at most $O(m_{max}^{1/3})$.
    Since $U'$ consists of \hd that remained \hd after the restart, the cost of
    $U'$ informing all its \hd neighbors that it left the MIS is $O(1)$ rounds
    and $O(\mmax^{2/3})$ messages.
            
    The remaining part of~\cref{alg:high} finds an MIS for the set of vertices
    in $U$.
    First,
                finding $V'_H$ requires $O(1)$ rounds and $O(m^{2/3})$ messages. This is
    due to the fact that the number of edges in the induced subgraph is
    $O(m^{2/3})$ and the diameter is $O(1)$.
    Then, the last $4$ steps of the remaining part of the
    algorithm can be performed in
$O(D\log^2 n)$ rounds where $D$ is the distance to the
leader by a modified version of Theorem 1.5 of~\cite{CPS20} (\cref{lem:cps}).
Since the distance to the leader in this case is $\leq 4$,
the number of rounds necessary
by our algorithm is $O(\log^2 n)$. Finally, the number of messages sent by
the vertices is upper bounded by high-degree vertices which enter the MIS and must inform
its high-degree neighbors. Since each high-degree vertex attempts
to enter the MIS at most $O(\log^2 n)$ times and enters the MIS at most once,
the number of messages sent by the high-degree vertices is bound by $O(m^{2/3}\log^2 n)$, equal
to the number of edges between all high-degree vertices times $O(\log^2 n)$.
Because both~\cite{CPS20} and~\cite{GGR20} are implemented in
the \textsc{Congest} model, our algorithm also can be run in small bandwidth.
\end{proof}

\begin{corollary}\label{cor:high-2-runtime}
\cref{alg:find-subgraph} requires $O(\log^5 n)$ amortized rounds and
$O(\mmax^{2/3}\log^5 n)$ amortized messages.
\end{corollary}

The proof of the above corollary is almost identical to the proof
of~\cref{lem:high-deg-computation}.
Since~\cref{alg:find-subgraph} is a subroutine of~\cref{alg:high} (except for
the usage of~\cite{GGR20}), the proof of
the above corollary immediately follows.

\begin{lemma}\label{lem:edge-insert-high-low}
    Given an edge insertion, $(u, v)$,
    between a \ld vertex $u$ and a \hd vertex $v$,
            restoring an MIS in the graph requires
    $O(\log^2 n)$ amortized rounds and $O(m_{max}^{2/3}\log^2 n)$ amortized
    messages.
\end{lemma}

\begin{proof}
    Suppose, first, that $u$ is not in the MIS, then, regardless of whether $v$
    is in the MIS, nothing happens. Then, suppose $u$ is in the MIS. If $v$ is
    not in the MIS, nothing happens. However, if $v$ is in the MIS, then we need
    to run~\cref{alg:high} with $u$ as the leader. Computing $U$ requires $O(1)$
    rounds and $O(\mmax^{2/3})$ messages since $v$ just needs to send a
    message to each of its high-degree neighbors that it left the MIS. The
    vertices $w \in N_{high}(v)$
    that do not have neighbors in the MIS compose $U$. Then,
    by~\cref{lem:high-deg-computation},~\cref{alg:high}
    runs in $O(\log^2 n)$ rounds
    and $O(\mmax^{2/3}\log^2 n)$ messages.
\end{proof}

\begin{lemma}\label{lem:edge-del-high-low}
    Given an edge deletion, $(u, v)$, between a \ld vertex $u$ and a \hd vertex
    $v$, restoring an MIS in the graph requires $O(1)$ amortized rounds
    and $O(\mmax^{2/3})$ amortized messages.
\end{lemma}

\begin{proof}
    If $u$ is not in the MIS, nothing happens. If $u$ is in the MIS, then $v$
    needs to check whether it has any neighbors remaining after the deletion
    that are in the MIS. If not, we need to add $v$ to the MIS and inform all
    its \hd neighbors. We first perform~\cref{alg:clean vertices} to restart
    $N_{high}(v) \cup \{v\}$. This requires $O(1)$ amortized rounds and
    $O(\mmax^{2/3})$ amortized messages by~\cref{lem:m-max-runtime}.
    For $v$ to inform all its
    \hd neighbors after the restart requires $O(1)$ rounds and $O(\mmax^{2/3})$
    messages.
\end{proof}

Finally, we give the runtimes of edge insertions and deletions between two \hd
vertices.

\begin{lemma}\label{lem:hd-insert}
    Given an edge insertion, $(u, v)$, between two \hd vertices, restoring an
    MIS in the graph requires $O(\log^2 n)$ amortized rounds and
    $O(\mmax^{2/3}\log^2 n)$ amortized messages.
\end{lemma}

\begin{proof}
    If at most one of $u$ or $v$ is in the MIS, nothing happens. Otherwise,
    one of the two vertices must leave the MIS. w.l.o.g. suppose $v$ leaves the MIS.
    We must run~\cref{alg:high} with $v$ as the leader. This requires $O(\log^2
    n)$ amortized rounds and $O(\mmax^{2/3}\log^2 n)$ amortized messages
    by~\cref{lem:high-deg-computation}.
\end{proof}

\begin{lemma}\label{lem:hd-delete}
    Given an edge deletion, $(u, v)$, between two \hd vertices, restoring an MIS
    in the graph requires $O(1)$ amortized rounds and $O(\mmax^{2/3})$ amortized
    messages.
\end{lemma}

\begin{proof}
    If none of $u$ and $v$ is in the MIS, nothing happens. If w.l.o.g. $u$ is in
    the MIS, then if $v$ does not have any neighbors in the MIS, it must enter
    the MIS. We run~\cref{alg:clean vertices} on $\{v\} \cup N_{high}(v)$ which
    costs $O(1)$ amortized rounds and $O(\mmax^{2/3})$ amortized messages
    by~\cref{lem:m-max-runtime}. Then, $v$ informs all its \hd neighbors it is
    in the MIS in $O(1)$ round and $O(\mmax^{1/3})$ messages.
\end{proof}

\cref{lem:insertion-low,lem:deletion-low,lem:edge-del-high-low,lem:edge-insert-high-low,lem:hd-delete,lem:hd-insert} immediate give the following theorem.
\fi

\begin{restatable}{theorem}{mmaxfinal}\label{thm:mmax-main}
    There exists a deterministic algorithm in the \congest model that maintains
    an MIS in a graph $G = (V, E)$ under edge insertions/deletions in $O(\log^2
    n)$ amortized rounds and $O(\mmax^{2/3}\log^2 n)$ amortized messages per
    update.
\end{restatable}

%% file: m-avg.tex
\subsection{\texorpdfstring{Extending from $m_{max}$ to the Average Number of Edges $m_{avg}$}{}}\label{sec:m-avg}

Thus far, our message and round complexity have been in terms of the
maximum number of edges present in the graph at any time.
In this section, we extend our results to the case of the
average number of edges in the graph $m_{avg}$, such that $m_{avg}$ is the
average number of edges throughout \emph{all} of the update sequence.
We do this by dividing the update sequence into \emph{phases},
such that if the number of edges in the beginning of phase $i$ is
denoted by $m_i$ then the length of phase $i$ is
$m_i^{2/3}$ updates (if $m_i^{2/3} < 1$ we assume we have one update at this phase).
In each phase, we obtain $\tilde{O}(1)$ amortized rounds
and $\tilde{O}(m_i^{2/3})$ amortized messages.
This result, using H\"{o}lder's inequality, will imply
that the amortized message complexity is $\tilde{O}(m_{avg}^{2/3})$ (see ~\cref{thm:holder} for more details).

To obtain this bound, we assume that
each edge insertion/deletion is tagged
with a \emph{timestamp} indicating the number of edge
updates (including itself) that have occurred since
the beginning. We assume
that when an edge update $(u, v)$ occurs, both $u$ and $v$
receive the timestamp $t_{(u, v)}$ associated with
the update. Using this information, we modify our algorithms
above for handling updates in the $m_{avg}$ case.

Each vertex $v$ stores two additional pieces of information, $S_v$
indicating the number of edges it thinks are in the graph
and $t_v$ indicating the timestamp of the last time $S_v$ was updated.
Initially $S_v$ is set to $0$ when the graph is empty.
If an edge insertion adjacent to $v$ occurs, and after the update $\deg(v)>S_v$, $v$ updates $S_v$ such that $S_v \geq \deg(v)$ is always satisfied; $t_v$ is updated
with the timestamp of the edge insertion that caused $S_v$ to be updated.
$S_v$ is \emph{not} updated
when an edge deletion that is adjacent
to $v$ occurs (if no restart algorithm is called). $S_v$ is also updated with a new value
when~\cref{alg:clean vertices} runs on $v$. When~\cref{alg:clean vertices}
runs on $v$, $v$ stores $S_v = S$ to be its new estimate of the number
of edges in its sub-component and $t_v$ to be the timestamp of the update
that caused~\cref{alg:clean vertices} to be called.

\paragraph{Updated main algorithm (based upon \cref{sec:main} with some minor changes}
Consider an edge update $(u,v)$ (insertion or deletion) such that
at least one vertex from $u$ and $v$ is \ld, w.l.o.g. say $u$ is \ld. In this
case, if $u$ needs to send a message to its neighbors, $u$ first
checks whether $t_{(u,v)} - t_u > S_u/4$ where $t_{(u,v)}$ is
the update timestamp it received. If $t_{(u,v)} - t_u > S_u/4$
then $u$ henceforth considers itself a \hd vertex. Therefore,
$u$ notifies all its neighbors about its movement to $V_H$
and performs \cref{alg:low-neighbor} if necessary. If $u$ was in
the MIS, every neighbor of $u$ that wakes up (say $w$) and needs
to enter the MIS, receives $t_{(u,v)}$ from $u$ and first checks
whether $t_{(u,v)} - t_w > S_w/4$ (notifying all of $u$'s
neighbors the timestamp requires
$O(\log n)$-bit messages assuming the timestamps resets to $0$ after
sufficiently many updates). If so, it follows the same
process as $u$. Note that more than one vertex can move
from $V_L$ to $V_H$ in a single update. However, we will prove
later that we can afford to pay for those movements.
Once a vertex moves to $V_H$, it can return to $V_L$ only by a restart,
which means that even if $w$ moved to $V_H$ but $\deg(w) \le deg'(w)$,
$w$ will not return to $V_L$ until a restart.
The rest of the process occurs as we described in the main algorithm.

\subsubsection{Updated restart procedure}
Our restart procedure works like the restart procedure given in~\cref{sec:restart}
except for the follow change: each vertex that
moves from $V_H$ to $V_L$, say $w$, not only updates $deg'(w)$ but
also saves $S_w \leftarrow S$ and updates $t_w \leftarrow t_j$, where
$t_j$ is the current update's timestamp.
Since we changed the definition of $m$, we also need to
update \cref{thm:restart} and \cref{lemma1} to be the following.

\begin{invariant}\label{inv:restart-low-deg}
    Any vertex adjacent to an edge update that happens during phase $i$ and is low-degree
    after our updated restart has degree $O(m_{i}^{2/3})$.
\end{invariant}

\begin{invariant}\label{inv:restart-high-deg}
    Any vertex adjacent to an edge update that happens during
    phase $i$ and is high-degree after our updated
    restart has at most $O(m_{i}^{1/3})$ high-degree vertices in its
    neighborhood when participating in~\cref{alg:high}.
\end{invariant}
\iflong
\begin{lemma}\label{lem:remain-avg}
After the updated \cref{alg:clean vertices} ends (at phase $i$), the number of vertices that
remain in $G'$ is at most $2m_i^{1/3}$.
\end{lemma}

\begin{proof}
As in \cref{thm:restart}, the number of vertices that remain in $G'$
is at most $S^{1/3}$. Note that $S = \sum_{v \in V'} \deg(v) \le 2m$
where $m$ is the current number of edges. Since we can bound $m$
by $m \le m_i + m_i^{2/3} \le 2m_i$, we obtain $S \le 2m \le 4m_i$
and therefore $S ^{1/3} \le 4^{1/3}m_i^{1/3} \le 2m_i^{1/3}$
\end{proof}

Using~\cref{lem:remain-avg}, we can prove our~\cref{inv:restart-high-deg} below.

\begin{lemma}\label{lem:inv-high-deg}
    \cref{inv:restart-high-deg} holds after all updated restarts.
\end{lemma}

\begin{proof}
    By definition of~\cref{alg:clean vertices}, the set of vertices in $G'$ are
    ones that participate in~\cref{alg:high}. Hence, by~\cref{lem:remain-avg}, the
    number of vertices in the neighborhood of each high-degree vertex is
    $O(m_{i}^{1/3})$.
\end{proof}

\begin{lemma}\label{lem:updated-d'v}
    If a vertex $v \in V'$ moves from $V_H$ to $V_L$, then $\deg(v) \leq 3m_i^{2/3}$.
\end{lemma}

\begin{proof}
By~\cref{alg:clean vertices}, we obtain that if a vertex $v \in V'$
moves from $V_H$ to $V_L$ then $\deg(v) < S^{2/3}$ . Since $S \le 2m \le 4m_i$,
we conclude that if $v \in V'$ moves from $V_H$ to $V_L$, then
$\deg(v) < S^{2/3} \le (4m_i)^{2/3} \le 3m_i^{2/3}$.
\end{proof}
\fi

Using~\cref{lem:updated-d'v}, we can prove our~\cref{inv:restart-low-deg} below.

\begin{lemma}\label{lem:inv3-holds}
    \cref{inv:restart-low-deg} holds after all updated restarts.
\end{lemma}

\begin{proof}
Suppose $v$ is a vertex adjacent to an edge update $(u,v)$ at phase $i$, and after the update $v \in V_L$. If $v \in V_L$ before the update, then after the update $v$ checks whether $\deg(v) > deg'(v)$ and whether $t_{(u,v)}- t_v \ge S_v/4$.
Since $v$ stays in $V_L$, it means that $\deg(v) \le deg'(v)$ and $t_{(u,v)}- t_v < S_v/4$. $deg'(v)$, $t_v$, and $S_v$ were updated during the last time $v$ moved from $V_H$ to $V_L$.
By the proof of \cref{lemma1} and according to~\cref{alg:clean vertices}, we obtain that when $v$ moved from $V_H$ to $V_L$ it updated $deg'(v)$ to be $2\deg(v)$ which is at most $2S^{2/3}$ and updated $S_v = S$.
By the assumption, there were at most $S_v/4$ updates since the last restart and the number of edges in the graph during the restart was at least $S_v/2$. Therefore, the number of edges in the graph is at least $S_v/4$.
In addition, since the update occurred at phase $i$, we get that $S_v/4 = O(m_i)$. By the fact that $\deg(v) < deg'(v) < 2S_v^{2/3}$ we get that $\deg(v) = O(m_i^{2/3})$.
If $v \in V_H$ before the update, then by ~\cref{lem:updated-d'v} we get that $\deg(v) = O(m_i^{2/3})$.
\end{proof}

Note that \cref{lem:d'v} is not true anymore because $m$ can decrease
significantly throughout phases but $deg'(v)$ can stay the same for more than one
phase, and in that case $deg'(v)$ could be more than $m_i$. Hence, the timestamp that a vertex $v$ receives during its movement from $V_H$ to $V_L$ is critical.

\begin{restatable}{lemma}{avgconst}\label{lem:avg-const}
    The average number of \ld vertices that enter or leave the MIS during an update is constant.
\end{restatable}
\iflong
\begin{proof}
In order to prove this lemma, we show that only a constant number of \ld vertices leave the MIS during an update.
A \ld vertex $u$ leaves the MIS for three reasons:
\begin{itemize}
    \item an edge insertion $(u,v)$ occurs and $\deg(u) > deg'(u)$ after the insertion.
    \item an edge update $(u,v)$ occurs and $t_{(u,v)}- t_u > S_u/4$
    \item an edge insertion $(u,v)$ occurs such that both $u$ and $v$ are \ld
        vertices and both are in the MIS.
\end{itemize}
In the first and the second cases, at most $2$ vertices might move to $V_H$ and leave the MIS ($u$ and $v$); that could happen if for both at least one of the two conditions is met.
In the third case, if none of the first two conditions is met, only one of $u$ and $v$ would leave the MIS.

Note that it is \emph{not} the case that only vertices that are adjacent to an
edge update can move to $V_H$. However, the other \ld vertices that can move to
$V_H$ are not in the MIS. Only vertices that need to make an action (either
enter or leave the MIS) check their timestamp and their degree. In addition,
those vertices first check their timestamp and only after that enter or leave
the MIS (if it is still necessary). When a vertex $u$ leaves the MIS, all its
\ld neighbors might enter the MIS, but before they enter they first check their
timestamp; therefore, if they leave $V_H$ to become part of $V_H$,
they first leave $V_H$ and enter the MIS as a \hd vertex later if necessary.
So, for these vertices, we ensure that they don't enter the MIS, and then
immediately need to leave the MIS because they move to $V_H$.

Overall, we get that at most $2$ \ld vertices leave the MIS during an update.
Thus, the average number of \ld vertices that enter the MIS during an update is
also constant since a vertex can only enter the MIS if it is not currently in
the MIS.
\end{proof}
\fi

\subsubsection{Updated analysis of the Restart Procedure for $\mavg$}
\begin{theorem}\label{thm:phase_analysis}
During phase $i$ for some $i\ge 0$, the amortized round complexity is
$\Tilde{O}(1)$ and the amortized message complexity is $\Tilde{O}(m_i^{2/3})$
for each update.
\end{theorem}

To prove \cref{thm:phase_analysis}, we divide the restart procedure into two
kinds of restarts:
\begin{itemize}
    \item \textbf{Heavy restart:} a restart that happens during phase $i$ in
        which the estimated number of edges in the graph $S > 4m_i^{2/3}$.
    \item \textbf{Light restart:} a restart in which the estimated number of
        edges in the graph $S\le 4m_i^{2/3}$.
\end{itemize}

\begin{lemma}\label{lem:light-restart}
    A light restart takes $O(1)$ amortized rounds and $O(m_i^{2/3})$
    messages.
\end{lemma}

\begin{proof}
Since the changes we made to our restart procedure is
only to save two more known values
for each vertex that moves from $V_H$ to $V_L$, the running time of the restart
did not change. Denote by $V'_M$ the vertices that moved from $V_H$ to $V_L$
during the restart and entered the MIS after the restart. We know
from~\cref{subsec:restart-analysis} that the restart procedure takes at most
$O(|V'_{M}|)$ rounds and $O(|E'| + S')$ messages. Since $|E'| +
S' \le S \le 4m_i^{2/3}$ we
obtain that the number of messages is at most $O(m_i^{2/3})$.
By~\cref{lem:avg-const} we know that the average number of vertices that enter
the MIS in an update is constant, therefore, the size of $V'_M$ is constant on
average and the number of rounds is amortized $O(1)$.
Note that this restart might cause other restarts in a chain (as we saw in the
$m_{max}$ case), we would show later on that we can handle this case also.
\end{proof}

The difficulty with our update procedure for $\mavg$ is that trivially
heavy restarts can cost more than $O(m_i^{2/3})$ worst-case,
and since we cannot bound the number of heavy restarts during a phase, we might
need to pay a lot for heavy restarts during phase $i$. In the next lemma, we
prove that a vertex $v$ would participate in a heavy restart only if it could
``pay'' for it.

\begin{restatable}{lemma}{numrestarts}\label{lem:num-of-restarts}
    If a vertex $v$ participated in a heavy restart $r_a$ during phase $i$ and
    moved from $V_H$ to $V_L$ during $r_a$, it can move back to $V_H$ and
    participate in another restart $r_b$ during phase $i$ only if its degree
    $\deg(v)$ was doubled after $r_a$.
\end{restatable}
\iflong
\begin{proof}
    Suppose $v$ is a vertex that moved to $V_L$ after the restart $r_a$.
In order for $v$ to move back to $V_H$
during an update $j$ that occurs in phase $i$ after $r_a$, one of
the following two conditions needs to be met:
\begin{itemize}
    \item $\deg(v) > deg'(v)$
    \item $t_v - t_j > S_v/4$
\end{itemize}
 where $t_j$ is the timestamp of update $j$. Since $t_v$ was updated during
 $r_a$, and $r_a$ and $j$ both occurred during phase $i$, we get that $t_v - t_j
 < m_i^{2/3}$. In addition, $S_v = S_{r_a} > 4m_i^{2/3}$ because $r_a$ is a
 heavy restart, which means that during any update $j$ in phase $i$, $t_v - t_j
 < m_i^{2/3} < S_v/4$, so the second condition would not occur during phase $i$.
 Therefore, $v$ can participate in another restart only if $\deg(v) >deg'(v)$ and since
 $deg'(v) = 2d^{r_a}_v$ (where $d^{r_a}_v$ is the degree of $v$ during the restart
 at $r_a$), then $\deg(v)$ must be at least $2d^{r_a}_v$, which means the degree of $v$ was doubled.
\end{proof}
\fi

\begin{restatable}{lemma}{heavyrestartcost}\label{lem:heavy-restart-cost}
    During phase $i$, the cost of a heavy restart is $O(1)$ amortized
    rounds and $O(m_i^{2/3})$ amortized messages.
\end{restatable}
\iflong
\begin{proof}
Overall, we get that in phase $i$ each heavy restart costs $O(|E'| +
S')$ messages which is at most $O(m_i^{2/3} + S')$. We use the same
argument as in the proof of~\cref{lem:m-max-runtime} to ignore the runtime
contribution of set $L$ to $|E'|$.
The inequality follows because $m_i$ can grow to at most $m_i + m_i^{2/3}$
during phase $i$. Since $S' = \sum_{v\in V_{low}} \deg(v)$ (where $V_{low}$
is the group of vertices that moved from $V_H$ to $V_L$ during the restart), we
can pay for each such vertex $v \in V_{low}$ if its degree doubled after the
restart.

By \cref{lem:num-of-restarts}, each vertex $v \in V_{low}$ that did not
double its degree after the heavy restart will not participate in another
restart during phase $i$. The number of vertices that does not pay for a heavy
restart during phase $i$ is at most $n$ and therefore the cost of those vertices
is at most $O(\sum_{v\in V} \deg(v))$; however, since $\deg(v)$ can change
during phase $i$, we will calculate the cost according to $deg^{max}(v)$ which is
the maximum degree of $v$ during phase $i$. So the maximum cost we pay for
vertices that participated in a heavy restart but did not pay their part of the
restart by doubling their degree is $O(\sum_{v\in V} deg^{max}(v)) =
O(\sum_{v\in V} \deg(v) + m_i^{2/3}) = O(m_i + m_i^{2/3}) =
O(m_i)$, so if we charge on each update during phase $i$ another
$O(m_i^{1/3})$ messages, we would pay at the end of phase $i$ for all
those vertices and get $O(m_i^{2/3})$ amortized message complexity
during this phase.
The round complexity of a heavy restart is at most $O(|V'_M|)$, where
again $V'_M$ is the vertices that moved from $V_H$ to $V_L$ during the restart
and entered the MIS after the restart.
By~\cref{lem:avg-const} we get that the average number of \ld vertices that
enter the MIS in an update is constant, therefore, the average size of $V'_M$ is
constant and the restart costs amortized $O(1)$ rounds.
\end{proof}
\fi
By \cref{lem:light-restart}, during phase $i$, we know that a light restart
costs $O(1)$ amortized rounds and $O(m_i^{2/3})$ messages and by
\cref{lem:heavy-restart-cost} we know that a heavy restart
also costs $O(1)$ amortized rounds and $O(m_i^{2/3})$ amortized
messages.

The last difficulty we encounter is
that one restart can cause a chain of restarts. However, recall that after a
restart only a vertex that moved from $V_H$ to $V_L$ and entered the MIS can
cause another restart. Since the average number of \ld vertices that enter
the MIS during an update is constant (by~\cref{lem:avg-const}), we get that the
average number of restarts during a phase is also constant and by our lemmas
above obtains our desired costs.

Overall, the payment for all the restarts during a phase is $O(1)$
amortized rounds and $O(m_i^{2/3})$ amortized messages.

\begin{restatable}{lemma}{movement}\label{lem:movement}
    The total cost of moving vertices from $V_L$ to $V_H$ is amortized $O(1)$ rounds and $O(m^{2/3})$ messages.
\end{restatable}
\iflong
\begin{proof}
Let $\Tilde{V}$ be all the vertices $v \in V$ with degree $\deg(v) > 4m_i^{2/3}$ at
the beginning of phase $i$. Note that each such vertex $v \in \Tilde{V}$ can
move to $V_H$ only once during phase $i$. That is because according to
\cref{lem:updated-d'v} a vertex $v \in V_H$ moves to $V_L$ only if $\deg(v) <
3m_i^{2/3}$, and since each vertex $v \in \Tilde{V}\cap V_H$ has degree $\deg(v) >
4m_i^{2/3}$, it can only decrease during phase $i$ to more than $3m_i^{2/3}$,
which is not enough for $v$ to move back to $V_L$.
Since the sum of the degrees of vertices in $\Tilde{V}$ is $O(m_i)$, we can pay for
the movement of those vertices by charging another $O(m_i^{1/3})$ messages in
each update.

Let's take a look on the vertices in $V \setminus \Tilde{V}$; each such vertex $v$ has degree $\deg(v) \le 4m_i^{2/3}$, so the cost of the movement of $v$ is at most $O(1)$ rounds and $O(m_i^{2/3})$ messages. However, the movement of $v$ might cause the movement of other vertices.
Recall that only vertices that want to make an action (enter or leave the MIS)
first check their stored timestamp and their degree and move to $V_H$ if necessary.
By~\cref{lem:avg-const}, the average number of such \ld vertices is constant.
Therefore, the total cost of the movement of vertices from $V \setminus \Tilde{V}$ is $O(1)$ amortized rounds and $O(m_i^{2/3})$ amortized messages.
\end{proof}
\fi

By~\cref{lem:movement} and since a vertex can from $V_H$ to $V_L$ only through a
restart, we get that overall the running time of the movements in phase $i$ is $O(1)$ amortized rounds and $O(m_i^{2/3})$ amortized messages.

%% file: analysis-avg.tex
\subsubsection{Analysis of the Main Algorithm under
$m_{avg}$}\label{sec:analysis-avg}
\iflong
Using our updated restart procedure for $\mavg$ provided in~\cref{sec:m-avg}, we
provide an updated analysis of our main algorithm in terms of $\mavg$ in this
section.

\begin{restatable}{lemma}{deletionlow}\label{lem:deletion-avg-low}
    Given an edge deletion during phase $i$ between two low-degree vertices $(u, v)$
    (where $u$ and $v$ are low-degree after any restart) where w.l.o.g.
    $u$ is in the MIS and $c_v = 0$ after the deletion, $v$ adds itself
    to the MIS after $O(1)$ rounds of communication and after sending
    $O(m_{i}^{2/3})$ messages.
\end{restatable}

\begin{proof}
    By~\cref{inv:restart-low-deg}, the degree of $v$ after the restart is
    $O(m_{i}^{2/3})$. Hence, by our procedure $v$ enters the MIS after sending a
    message informing all its neighbors that it is entering the MIS. Hence,
    because $deg(v) = O(m_i^{2/3})$, $v$ adds itself to the MIS after $O(1)$ rounds
    of communication and after sending $O(m_i^{2/3})$ messages.
\end{proof}

\begin{lemma}\label{lem:insertion-avg-low}
    Given an edge insertion during phase $i$ between low-degree vertices $(u,
    v)$ (where $u$ and $v$ are low-degree after any restart)
    such that both are in the MIS,~\cref{alg:low-neighbor} requires $O(1)$
    amortized rounds and $O(m_{i}^{2/3})$ amortized messages using
    $O(\log n)$-bit messages.
\end{lemma}

\begin{proof}
    By~\cref{inv:restart-low-deg}, the degree of $u$ and $v$ after the restart
    is $O(m_{i}^{2/3})$. Therefore, w.l.o.g. if $u$ leaves the MIS then it notifies
    all its neighbors with $O(1)$ rounds and $O(m_i^{2/3})$ messages. Each of
    $u$'s \ld neighbors that needs to enter the MIS, checks first if it needs to
    move to $V_H$. If so, it moves to $V_H$ and notify all its neighbors (we
    analyze the cost of the movement above).
    If not, by the proof of~\cref{lem:inv3-holds} we get that its degree is at
    most $O(m_i^{2/3})$ and therefore it pays $O(1)$ rounds and $O(m_i^{2/3})$
    message to notify all its neighbors.
    By~\cref{lem:avg-const} we know that the average number of \ld vertices that
    enter the MIS during an update is constant.
    Therefore, the amortized cost of this update is $O(1)$ rounds and
    $O(m_i^{2/3})$ messages.
                \end{proof}

\begin{lemma}
    Determining the set of high-degree neighbors to enter the MIS after an edge insertion
    between a high-degree node in the MIS and a low-degree node in the MIS requires
    $O(\log^2 n)$ rounds and $O(m_{i}^{2/3}\log^2 n)$ messages
    using~\cref{alg:high}.
\end{lemma}

\begin{proof}
    By~\cref{inv:restart-high-deg} we get that after the restart, the number of \hd vertices that participate in~\cref{alg:high} is at most $O(m_i^{1/3})$. Obtaining the same proof of~\cref{lem:high-deg-computation} with $m_i$ instead of $m_{max}$ would suffice.
\end{proof}

The rest of the lemmas, \cref{lem:edge-del-high-low,lem:hd-insert,lem:hd-delete}, follow
immediately for phase $i$ from our invariants and by replacing $\mmax$ by $m_i$.
We show as our last step that we can translate our costs from $m_i$ to $\mavg$.

\begin{lemma}\label{thm:holder}
If there are $t$ updates, where in phase $i$ the average number of messages
during each update is $\Tilde{O}(m_i^{2/3})$, such that $O(m_i)$ is the maximum
number of edges in this phase, then the number of messages during each update is
$\Tilde{O}(m_{avg}^{2/3})$.
\end{lemma}

\begin{proof}
   In order to prove the theorem, we would like to prove that the average number of messages throughout all the updates is less than $\Tilde{O}(m_{avg}^{2/3})$.
   Since the average number of messages during each update in phase $i$ is
   $\Tilde{O}(m_i^{2/3})$ and the number of edges during this phase is $O(m_i)$,
   we can just refer to the number of sent messages to resolve an update $j$ as
   $(c_jm_j)^{2/3}\log ^2n$ (for some constant $c_j$) and to the number of edges in
   this update as  $c_jm_j$, respectively.
   So we need to prove that:
   \begin{align}
       \frac{\sum_{j=1}^t(c_jm_j)^{2/3}\log^2 n}{t} \le
       \left(\frac{\sum_{j=1}^tc_jm_j}{t}\right)^{2/3}\log^2 n
       \label{eq:avg}
   \end{align}
   and by the fact that
   \begin{align}
       \left(\frac{\sum_{j=1}^tc_jm_j}{t}\right)^{2/3} = m_{avg}^{2/3}
       \label{eq:mavg}
   \end{align}
   we would prove our theorem.
   Note that~\cref{eq:avg} is equal to:
   \begin{align}
       \sum_{j=1}^t(c_jm_j)^{2/3} \le \left(\sum_{j=1}^tc_jm_j\right)^{2/3}t^{1/3}
       \label{eq:avg2}
   \end{align}
   In order to prove~\cref{eq:avg2} we need to use H\"{o}lder's inequality:
   \begin{align}
       \label{eq:holder}
       \sum_{k=1}^n|x_ky_k| \le  \left(\sum_{k=1}^n|x_k|^p\right)^{\frac{1}{p}} \left(\sum_{k=1}^n|y_k|^q\right)^{\frac{1}{q}}
   \end{align}
   such that $(x_1,...,x_n), (y_1,...,y_n)\in \mathbb{R}^n$, $p,q \ge 1$ and $\frac{1}{p} + \frac{1}{q} = 1$.
   By choosing $p = \frac{3}{2}, q = 3$, and for each $1 \le j \le t$ $y_j = 1, x_j = (c_jm_j)^{2/3}$ we get that:
      \begin{align}
       \sum_{j=1}^t|(c_jm_j)^{\frac{2}{3}}| \le \left(\sum_{j=1}^t|(c_jm_j)^{\frac{2}{3}}|^{\frac{3}{2}}\right)^{\frac{2}{3}} \left(\sum_{j=1}^t|1|^3\right)^{\frac{1}{3}}
   \end{align}
   which is exactly~\cref{eq:avg2}.
\end{proof}

Using the above lemmas, we obtain our final result for this section with respect
to $\mavg$.
\else
Using our updated restart procedure for $\mavg$ provided in~\cref{sec:m-avg}, we
provide an updated analysis of our main algorithm in~\cref{app:m-avg} terms of $\mavg$,
due to space constraints. The analysis then provides our main theorem in terms of $\mavg$.
\fi
\begin{restatable}{theorem}{mavgfinal}\label{thm:m-avg-final}
There exists a deterministic algorithm in the \congest model that maintains an
MIS in a graph $G = (V, E)$ under edge insertions/deletions in $O(\log^2 n)$
amortized rounds and $O(\mavg^{2/3}\log^2 n)$ amortized
messages per update.
\end{restatable}

%% file: appendix.tex
\section{\texorpdfstring{$O(D\log^2 n)$ round Static MIS algorithm in \congest}{}}\label{app:cps}

We describe the key components of~\cite{CPS20} that are necessary to
prove~\cref{lem:cps}. The algorithm provided in~\cite{CPS20} operates
under the \congest model. The algorithm consists of $O(\log n)$ phases
each of which requires $O(D \log n)$ rounds of communication
where $D$ is the diameter of the graph. Each phase is one round of
a modified and derandomized version of~\cite{Ghaffari16}. Consider phase $t$ of
the algorithm
where we only consider the vertices that are not yet removed from the graph
(i.e.\ the vertices that have not yet been added to the MIS nor are neighbors of
vertices in the MIS). Each vertex stores in its local memory a (large enough)
family of hash
functions and they choose a hash function based on a shared random seed that is
computed deterministically.

Let seed $Y = (y_1, \dots, y_{\gamma})$ be the $\gamma$ random variables that are
used to select a hash function from the hash function family. The value of $y_i$
in the seed is computed in the following way. Suppose $y_1 = b_1, \dots, y_{i
-1} = b_{i - 1}$ have already been computed. By Ghaffari's algorithm, a vertex can
be in one of two types of \emph{golden-rounds}. If a vertex $v$ is in \emph{golden
round-1}, $v$ uses the IDs of its neighbors (that are not yet removed) and some
additional information from its neighbors to compute
\begin{align*}
E(\psi_{v, t}|Y_{i, b}) = Pr[m_{v, t} = 1| Y_{i, b}] - \sum_{u \in N(v)}
Pr[m_{u, t} = 1 | Y_{i, b}]
\end{align*}
where $Pr[m_{v, t} = 1 | Y_{i, b}]$ is the
conditional probability that $v$ is marked in phase $t$.

For vertices in \emph{golden-round-2}, the vertex $v$ first finds a subset of the
neighbors $W(v) \subseteq N(v)$ satisfying
$\sum_{w \in W(v)} p_t(w) \in [1/40, 1/4]$. Then, let $M_{t, b}(u)$ be the
    conditional probability on $Y_{i, b}$ that $u$ is marked but none of its
    neighbors are marked. Let $M_{t, b}(u, W(v))$ be the conditional probability
    that another neighbor of $v$, $w \in W(v)$, other than $u$ is marked. Using
    these values, $v$ then computes
    \begin{align*}
        E(\psi_{v, t}|Y_{i, b}) = \sum_{u \in W(v)} Pr[m_{u, t} = 1 | Y_{i, b}] - M_{t, b}(u) - M_{t, b}(u, W(v)).
    \end{align*}
    Each vertex then sends $E(\phi'_{v, t}|Y_{i, b})$ (computed using
    $E(\psi_{v, t}|Y_{i, b})$) to its parent who sums up all the values they
    received. The leader receives the aggregate values $\sum_{v \in V'}
    E(\phi'_{v, t} | Y_{i, b} = 0)$ and $\sum_{v \in V'} E(\phi'_{v, t} | Y_{i,
    b} = 1)$ from its descendants and decides on $y_i = 0$ if $E(\phi'_{v, t} | Y_{i, b} = 0) \geq
    E(\phi'_{v, t} | Y_{i, b} = 1)$ and $y_i = 1$ otherwise. The leader then
    sends the bit to all its descendants.

    Once the vertices have computed $Y$, they can simulate phase $t$ of Ghaffari's
    algorithm. Each vertex that gets marked enters the MIS and removes themselves
    from the subgraph running the algorithm if none of their
    neighbors are marked. All neighbors of such vertices also remove themselves.
    The algorithm proceeds with this smaller subgraph.